\documentclass[12pt]{article}
\usepackage{tikz}
\usepackage{amssymb}
\usepackage{amsmath}
\usepackage{amsfonts}
\usepackage{array}
\usepackage{graphicx}
\usepackage{setspace}
\usepackage{amsthm}
\usepackage{bbm}
\usepackage[margin=1in]{geometry}

\newcommand{\black}{\mathbb{B}}
\newcommand{\white}{\mathbb{W}}
\newcommand{\lone}[1]{||#1||_1}

\newcommand{\Sy}{\mathcal{S}}
\newcommand{\One}{\textbf{1}}
\newcommand{\J}{\mathcal{J}}

\newcommand{\M}{\mathcal{M}}

\newcommand{\W}{\mathcal{W}}

\newcommand{\I}{\mathcal{I}}

\newcommand{\G}{\mathcal{G}}
\newcommand{\Q}{\mathcal{Q}}
\newcommand{\K}{\mathcal{K}}
\newcommand{\A}{\mathcal{A}}
\newcommand{\N}{\mathcal{N}}
\newcommand{\R}{\mathcal{R}}
\newcommand{\C}{\mathcal{C}}
\newcommand{\Zp}{\mathbb{Z}^+}
\newcommand{\Zgeqz}{\mathbb{N}}
\newcommand{\argmax}{\operatornamewithlimits{argmax}}

\newtheorem{MyThm}{Theorem}
\newtheorem{MyLem}[MyThm]{Lemma}
\newtheorem{MyCorr}[MyThm]{Corollary}
\newtheorem{MyDef}[MyThm]{Definition}
\newtheorem{MyExa}[MyThm]{Example}
\newtheorem{MyPro}[MyThm]{Proposition}

\numberwithin{MyThm}{section}

\begin{document}

\title{On Threshold Models over Finite Networks\thanks{This research is partially supported by an MIT Jacobs Presidential Fellowship, a Siebel Scholarship, a Xerox Fellowship, AFOSR grant FA9550-09-1-0420, and ARO grant W911NF-09-1-0556.}}
\author{Elie M. Adam, Munther A. Dahleh, Asuman Ozdaglar\thanks{All authors are with the Laboratory for Information and Decision Systems, Massachusetts Institute of Technology, Cambridge, MA 02139. %
(emails: eadam@mit.edu, dahleh@mit.edu, asuman@mit.edu)}}
\date{}

\maketitle

\begin{abstract}
We study a model for cascade effects over finite networks based on a deterministic binary linear threshold model. Our starting point is a networked coordination game where each agent's payoff is the sum of the payoffs coming from pairwise interactions with each of the neighbors. We first establish that the best response dynamics in this networked game is equivalent to the linear threshold dynamics with heterogeneous thresholds over the agents. While the previous literature has studied such linear threshold models under the assumption that each agent may change actions at most once, a study of best response dynamics in such networked games necessitates an analysis that allows for multiple switches in actions. In this paper, we develop such an analysis and construct a combinatorial framework to understand the behavior of the model. To this end, we establish that the agents behavior cycles among different actions in the limit and provide three sets of results.

We first characterize the limiting behavioral properties of the dynamics. We determine the length of the limit cycles and reveal bounds on the time steps required to reach such cycles for different network structures. We then study the complexity of decision/counting problems that arise within the context. Specifically, we consider the tractability of counting the number of limit cycles and fixed-points, and deciding the reachability of action profiles. We finally propose a measure of network resilience that captures the nature of the involved dynamics. We prove bounds and investigate the resilience of different network structures under this measure.
\end{abstract}

\newpage

\section{Introduction}
Networks intertwine with every aspect of our modern lives, be it through sharing ideas, communicating information, shaping opinions, performing transactions or delivering utilities. Explicitly, we may cite social networks, financial networks, economic networks, communication networks and power networks. Interactions over those many different types of networks require agents to coordinate with their neighbors. In economic networks, technologies that conform to the standards used by other related firms are more productive; in social networks, conformity to the behavior of friends is valuable for a variety of reasons. The desire for such coordination can lead to cascading effects: the adoption decision of some agents can spread to their neighbors and from there to the rest of the network. One of the most commonly used models of such cascading behavior is the linear threshold model introduced by Granovetter \cite{GRAN01} to explain a variety of aggregate level behaviors including diffusion of innovation, voting, spread of rumors, riots and strikes. The linear threshold model is also adopted as a model of cascading failures, propagation of diseases and opinion dynamics.

Most analyses of this model in the literature assume that one of the states adopted by the agents (represented by the nodes of a graph) is irreversible, meaning that agents can only make a single switch into this state and can never switch out from it. However, incurring this progressive property in behavior dilutes several perspectives of the dynamics: whereas some situations are best captured by such a variant, many others cannot be captured but by allowing players to revert back to previous actions. One motivation for example would be opinion dynamics in social network: in most situations a player changes opinions back and forth. Another motivation would be recovery within the context of cascading failures or infectious diseases. This said, the literature lacks a satisfactory characterization of the limiting properties of such a model.

In this paper, we consider a model of cascade effects based on binary linear threshold dynamics over finite graphs. We start from an explicit coordination game set over a finite undirected network. The payoff of each agent is the sum of the payoffs from {\it two player and two action coordination games} the agent plays pairwise with each of the neighbors (the action is fixed across all interactions). We then study the behavior induced by best response dynamics, whereby each agent changes the played action to that which yields highest payoff given the actions of the neighbors. We first establish that best response dynamics are identical to the dynamics traced by the linear threshold model with heterogeneous thresholds for the agents. However, crucially, actions can change multiple times. Thus, the dynamics of interest for the set of problems posed here cannot be studied using existing results and in fact have a different mathematical structure. The main contribution of this paper is to fully characterize these dynamics. We provide three sets of results.

	We establish that agent behavior cycles among different actions in the (time) limit, we term such limiting behavior as \emph{convergence cycles}. We characterize the length of those convergence cycles. Ultimately, we show that for any graph structure on the players, any threshold distribution over the players and any initial action configuration played by the players, the limiting behavior of the dynamics get absorbed into action configuration cycles of length at most two. In other words, at the limit, every agent either plays one action, never deviating, or keeps on switching actions at every time step. We then characterize the time required to reach convergence cycles, termed as \emph{convergence time}. Building on the framework set up, we show that there exists some positive integer $c$, such that given any graph structure on the players, any threshold distribution over the players and any initial action configuration, the dynamics reach a non-degenerate cycle or a fixed-point in at most $cn^2$ time steps where $n$ is the number of players. We mention that the work \cite{GOLE01} (in the literature on Cellular Automata) considers the same dynamics and proves similar bounds on the length of convergence cycles and a quadratic bound on convergence time. Nevertheless, our proof approach substantially differs in that we emphasize the combinatorial structure of the problem. Our key contribution in this respect lies in two transformations (later defined and termed as \emph{bipartite-expansion} and \emph{symmetric-expansion}) whereby the dynamics (over any initial structure) are extended over a bipartite graph structure where each player has an odd number of neighbors and a majority decision rule. The bipartite-expansion would allow us to decouple our parallel decision scheme to a sequential process, and the symmetric expansion acts as to reduce the proof (of convergence cycle bounds) to counting specific edges in the graph. Using this approach, the quadratic bound on convergence time naturally follows. These transformations would also allow us to answer further questions about fixed points and cycles, and improve the convergence time bound from quadratic to be uniformly not more than the size of the network whenever the graph in concern is either an even-length cycle graph or a tree.
	
	We then study the complexity of counting and decision problems that arise in this model. We are interested in characterizing the number of limiting states the system could get absorbed in. We begin by arguing that no `insightful' uniform upper-bound or lower-bound can be established. Considering only the case of a cycle graph, the number of fixed-points may vary at least from $2$ to $2^{n/3}$ depending on the threshold distribution. Instead, we turn to study how tractable it is to count the convergence cycles. We proceed to show that given a graph structure on the player and a threshold distribution over the players as input, the problems of counting the number of limiting configuration classes (i.e. either fixed-points or non-degenerate cycles), counting the number of fixed-points and counting the number of non-degenerate cycles are all \#P-Complete.
We further show that the problem of deciding whether an action configuration over the network is reachable along the dynamics is NP-Complete and that counting the number of action configuration preceding a \emph{reachable} action configuration is \#P-Complete.

	Finally, the resilience of networks to invasion by certain types of behavior (e.g., cascades of failures or spread of epidemics) is of central importance in the study of cascades over networks. For the new dynamics defined by our problem, we define a measure of resilience of a network to such invasion that captures the the minimal `cost of recovery' needed when the model is confronted with a perturbation in the agents' action profile. We prove achievable uniform lower-bounds and upper-bounds on the resilience measure, we compute the resilience measure of some network structures and provide basic insight on how different network structures affect this measure.

The paper is related to a large literature on network dynamics and linear threshold models (see e.g., \cite{MORR01}-\cite{BLUM01}). A number of papers in this literature investigate the question of whether a behavior initially adopted by a subset of agents (i.e., {\it the seed set}) will spread to a {\it large} portion of the network, focusing on the dynamics where agents can make a single switch to one of the behaviors. Morris \cite{MORR01}, while starting from a multi-switch version of the dynamics, studied without loss of generality the single-switch version to answer whether there exists a finite set of initial adopters (in an infinite network with homogeneous thresholds) such that the behavior diffuses to the entire network. In \cite{WATT01}, Watts derives conditions for the behavior to spread to a positive fraction of the network (represented by a random graph with given degree distribution) using a branching process analysis. Similarly, Lelarge \cite{LELA01} provides an explicit characterization of the expected fraction of the agents that adopt the behavior in the limit over such networks.

	Some works study the single-switch version of the linear threshold model (where nodes never switch back from one of the states) over deterministic finite graphs. Given an initial seed set of adopter, the work \cite{ACEM01} characterizes the final set of adopters in terms of {\it cohesive sets} where cohesion in social groups is measured by comparing the relative frequency of ties among group members to ties with non-members. The work in \cite{KLEI02} studies (e.g. in the context of viral marketing) how to target a fixed number of agents (and change their behavior) in order to maximize the spread of the behavior in the network in the (time) limit. It studies the (optimization) problem of maximizing the final set of adopter, under the constraint of picking $K$ initial adopters. It considers various models of cascade, shows that the optimization problem is NP-hard for the linear threshold model, and then provides an algorithm to find an approximation for the optimal set that achieves maximum influence. Although simulating the multi-switch version of the linear threshold model by a single-switch version is both feasible and insightful (see \cite{KLEI01} and \cite{KLEI02}), the results available (on the single-switch version) do not enable us to characterize the limiting behavior of the richer dynamics.

	In the context of network resilience, the recent paper \cite{BLUM01} adopts single-switch linear threshold dynamics as a model of failures in a network. This work defines a measure of network resilience that is a function of the graph topology and the distribution over thresholds and studies this measure for different network structures focusing on $d$-regular graphs (hence ignoring the effect of the degree distribution of a graph on cascaded failures). Here we provide a resilience measure that highlights the impact of heterogeneity in thresholds and degrees of different agents.

	Finally, noisy versions of best-response dynamics in networked coordination games were studied in \cite{YOUN01} and \cite{MONT01} (see also \cite{LIGG01} and \cite{DURR01} in the statistical mechanics literature). The random dynamics in these models can be represented in terms of Markov chains with absorbing states, and therefore do not exhibit the cyclic behavior predicted by the multi-switch linear threshold model studied in this paper.

The rest of the paper proceeds as follows. We begin by a description of the model in Section 2. We then proceed in Section 3 to describe the general behavioral rules of the dynamics. We branch out to characterize convergence cycles and convergence time in Sections 4 and 5. We study the complexity of counting and decision problems in Section 6. We finally propose the network resilience measure in Section 7, and conclude the paper in Section 8.
Moreover, certain (non-contiguous) subsets of sections form coherent logical entities. The paper allows readers to select specific topics without having to read through all the sections to ensure understanding. Convergence cycle lengths will be fully treated by the end of Section 4. Readers interested only in convergence time may read Sections 2, 3, 4 (up to 4.3, inclusive) and 5. Readers interested only in the complexity results may restrict their reading to Sections 2, 3, 4 (up to 4.3, inclusive) and 6. Finally, readers interested only in the resilience part may restrict their reading to Sections 2, 3 and 7.

\section{Model} 

The model section consists of two parts. The first sets up a networked coordination game, and characterizes the best-response dynamics in this game. The second proposes a natural extension to the dynamics: it imposes non-equal weights on the pairwise interactions among the players, and allows players to give weights to their own actions.

\subsection{The Primary Model} \label{Model_Dynamics}

We define a networked coordination game. For a positive integer $n$, we denote by $\I_n$ the set of $n$ players\footnote{We use the words player, agent, node and vertex interchangeably. We use the letters $i$ and $j$ to denote agents. We reserve the letter $n$ for the number of players in the game. If it is clear from the context to which set $X$ an element $x$ belongs to, we refrain from mentioning the set $X$ explicitly to simplify notation. Moreover, for any function $f$ with domain $\I_n$, we will denote $f(i)$ by $f_i$. In particular, for functions $q$, $k$ and $a$ with domain $\I_n$, $q(i)$, $k(i)$ and $a(i)$ are denoted $q_i$, $k_i$ and $a_i$ respectively.}. We define $\G_n$ to be the class of all connected undirected graphs $G(\I_n, E)$ defined over the vertex set $\I_n$, with edge set $E$.\footnote{For a graph $G$, we denote by $V(G)$ and $E(G)$ the vertex set and edge set respectively.} To be proper, $E$ is a relation\footnote{A (binary) relation $R$ on a set $A$ is a subset of $A{\times}A$. We use the notation $aRb$ to denote $(a,b) \in R$.} on $\I_n$, but for convenience (since the graph is undirected) we will consider the set $E$ to have cardinality exactly equal to the number of undirected edges. We denote an undirected edge in $E$ by $\{i,j\}$, and we abbreviate it to $ij$ when no confusion arises. For $G(\I_n,E)$ in $\G_n$, we use $\N_G(i)$ to denote the neighborhood of player~$i$ in $G$, i.e. $\N_G(i) = \{j~\in~\I_n : ij \in E \}$. We denote by $d_G(i)$ the degree of player~$i$ in $G$, i.e. the cardinality of $\N_G(i)$. We refer to $\N_G(i)$ and $d_G(i)$ respectively as $\N_i$ and $d_i$ when the underlying graph is clear from the context. We finally define $\Q_n$ to be the space of type distributions over the agents, specifically the set of maps from $\I_n$ into $[0,1]$. We refer to $q_i$ as the type of player $i$.

Let $\{\black,\white\}$ be a (binary) set of actions, where the symbols $\black$ and $\white$ may be identified with the colors \emph{black} and \emph{white}, respectively. Given a graph $G(\I_n,E)$ in $\G_n$ and a type distribution $q$ in $\Q_n$, each player $i$ in $\I_n$ plays one action $a_i$ in $\{\black,\white\}$. For $ij \in E$, we define the payoff received by agent $i$ when playing $a_i$ against agent $j$ playing $a_j$ to be
\begin{equation} \label{payoff}
g_{i,j}(a_i,a_j) =
\left\{
	\begin{array}{ll}
		q_i & \text{if } a_i = a_j = \white \\
		1-q_i & \text{if } a_i = a_j = \black\\
		0 & \text{if } a_i \neq a_j\\
	\end{array}. 
\right.
\end{equation}

The utility function of player $i$ is the sum of the payoffs from the pairwise interactions with the players in $\N_i$. Formally, when player $j$ plays action $a_j$, the utility function of player $i$ is given by:
\begin{equation} \label{sumPayoff}
u_i(a_i,a_{-i}) = \sum_{j \in \N_i}{g_{i,j}(a_i,a_j)},
\end{equation}
where $a_{-i}$ denotes the action profile of all players except~$i$.

We define $\A_n$ be the space of action profiles\footnote{We use the words action, assignment and color interchangeably, and use the words profile and configuration interchangeably.} played by the agents, specifically the set of maps from $\I_n$ into $\{\black,\white\}$. The players are assigned an initial action profile $\underline{a}$, we refer to $\underline{a}$ as the action profile of the players at time step $0$. For $T$ in $\Zp$,\footnote{We denote by $\Zgeqz$ the set of non-negative integers, and by $\Zp$ the set of positive integers.} every player best responds to the action profile of the players at time step $T-1$, by choosing the action that maximizes his utility function. We suppose that players play action $\white$ as a tie breaking rule. Formally we impose a strict order on $\{\white,\black\}$ such that $\min\{\white,\black\} = \white$. This tie breaking rule does not affect the behavior of the dynamics whatsoever. However, it does have a natural effect in the network resilience context and this effect will be discussed in Section 7. Suppose we denote by $a_{i,T}$ the action played by player $i$ at time $T$, then given an initial action configuration $\underline{a}$ in $\A_n$, for every player $i$, we recursively define:
\begin{align}\label{gameRule}
	a_{i,0} &= \underline{a}_i  \nonumber\\
	a_{i,T} &= \min \argmax_{a_i \in \{\white,\black\}} u_i(a_i,a_{-i,T-1}), \quad \text{for }  T \in \Zp.
\end{align}
where the $\min$ operator breaks ties. The following proposition provides a rule equivalent to that induced by the recursive definition in (\ref{gameRule}). This characterization is similar to that in \cite{MORR01}.

\begin{MyPro} \label{rule}
Let $\underline{a}$ be the initial action configuration, namely the action profile of the players at time step 0. For every positive integer $T$, player~$i$ plays action $\black$ at time step $T$ if and only if (strictly) more than $q_id_i$ neighbors of player $i$ played action $\black$ at time step~$T-1$.
\end{MyPro}

\begin{proof}
We substitute $u_i$ in (\ref{gameRule}) with the expressions in (\ref{payoff}) and (\ref{sumPayoff}), and get that player~$i$ plays action $\black$ at time $T$ if and only if
\begin{equation}
 \sum_{j \in \N_i}(1 - q_i)\One_{\{\black\}}(a_{j,T-1}) > \sum_{j \in \N_i}q_i\One_{\{\white\}}(a_{j,T-1}),  \nonumber
\end{equation}
where $\One_{\Gamma}(x) = 1$ if and only if $x\in\Gamma$. Equivalently, player~$i$ plays action $\black$ at time $T$ if and only if
\begin{equation}
 \sum_{j\in \N_i}\One_{\{\black\}}(a_{j,T-1}) > q_id_i. \nonumber
\end{equation}
The left-side term is essentially summing the number of neighbors of player $i$ playing action $\black$.
\end{proof}

As a technical clarification, we highlight the fact that every player is capable of switching actions both from $\white$ to $\black$ and $\black$ to $\white$. This contrasts a variant of the dynamics (extensively studied in the literature) where a player can switch out from only one of the actions.

\subsection{The Extension Model}
Our primary model in the best-response dynamics is such that every node does not take its own action into account and treats the payoffs from the pairwise interactions with equal weights. This corresponds on the part of player $i$ to an unweighted counting of the number of neighbors playing $\black$ at time $T-1$ to decide whether to play $\black$ at $T$ or not. Our model can take a more general form by allowing each node to play a coordination game with itself, and by assigning symmetric (with respect to the neighbors) weights on the payoffs. We will also allow those weights to be negative, and thus lose the monotonicity property induced by coordination to get richer dynamics. Neighbors linked by a negatively weighted edge have an incentive to mismatch their actions. In this respect, the induced networked game ceases to be (in general) a networked coordination game. Let $G(\I_n,E)$ be given. Suppose we assign for every $ij$ in $E$, a non-zero real weight $w_{ij}=w_{ji}$ and for every player $i$ in $\I_n$, a real weight $w_{ii}$. Given a $q$ in $\Q_n$, we extend the utility player $i$ gets to be the weighted sum of the payoffs from the pairwise interactions with the players in $\N_i \cup \{i\}$, specifically when player $j$ plays action $a_j$,
\begin{equation} \label{utility}
	u_i(a_i,a_{-i}) = \sum_{j \in \N_i\cup \{i\}}w_{ij}{g_{i,j}(a_i,a_j)}. \nonumber
\end{equation}
Again, we denote by $a_{i,T}$ the action played by player $i$ at time $T$. If we let $\underline{a}$ be the initial action configuration, namely $a_{i,0} = \underline{a}_i$, then for every positive integer $T$, player~$i$ plays action $a_{i,T} = \black$ at time step $T$ if and only if
\begin{equation}
	\sum_{j \in \N_i\cup \{i\}} w_{ij} \One_{ \{\black \} }(a_{j,T-1}) > \theta_i, \nonumber
\end{equation} 
where we define $\theta_i = q_i \sum_{j \in \N_i\cup \{i\}} w_{ij}$. The primary model described in the previous subsection is then an instance of this model where $w_{ij} = w_{ji} = 1$ for all edges $ij$ in $E$ and $w_{ii} = 0$ for all players $i$.

\section{Description of the Dynamics} \label{genbeh}

We proceed to provide a coarse description of the involved dynamics. We focus only on the primary model throughout this section; all propositions regarding the primary model may be naturally generalized to the extension model. We first highlight what is essential (in our model) for this section. We consider a \emph{finite} set of players $\I_n$ along with three mathematical objects $\G_n$, $\Q_n$ and $\A_n$. An element $G$ of $\G_n$ corresponds to the network structure imposed on the players, an element $q$ of $\Q_n$ refers to the type distribution over the players i.e. a function from $\I_n$ into $[0,1]$, and an element $a$ of $\A_n$ represents an action profile played by the players i.e. a function from $\I_n$ into $\{\white,\black\}$. The elements $G$, $q$ and $a$ interact as dictated by Proposition \ref{rule}.


\subsection{From Types to Thresholds}
 
In Proposition \ref{rule}, Player $i$ uses $q_id_i$ as a threshold to decide whether to play $\black$ or $\white$. The value of $q_id_i$ is non-necessarily an integer, however we may replace it by an integer without modifying the dynamics. To this end, we substitute the set $\Q_n$ by a set $\K_n$ and then modify the statement of Proposition \ref{rule}. We define $\K_n$ to be the space of threshold distributions over the agents, i.e. the set of maps from $\I_n$ into $\Zgeqz$. We refer the $k_i$ as the threshold of player $i$. We make a particular distinction between the word \emph{type} attributed to $\Q_n$ and the word \emph{threshold} attributed to $\K_n$. Given a pair $(G,k)$ with $k \in \K_n$, Proposition \ref{rule} generalizes as follows:
\begin{MyPro} \label{ruleInt}
Let $\underline{a}$ be the initial action configuration, namely the action profile of the players at time step $0$. For every positive integer $T$, player~$i$ plays action $\black$ at time step $T$ if and only if at least $k_i$  neighbors of player $i$  played action $\black$ at time step~$T-1$.
\end{MyPro}

The rule in Proposition \ref{ruleInt} supersets the rule in Proposition~\ref{rule}. For every $q$ in $\Q_n$ there exists a $k$ in $\K_n$ such that $q_id_i$ may be substituted with the integer $k_i$ for all~$i$ without changing the behavior of the players. Thus, a \emph{model} with types may always be simulated by a \emph{model} with integers. However, it is crucial to note that \emph{(strictly) more than} is replaced by \emph{at least}. Therefore, in this setting, the converse is not true i.e. not every \emph{model} with integers may be simulated by a \emph{model} with types. Indeed, Proposition \ref{rule} implies that playing $\black$ is never a best response for player $i$ if no player in $\N_i$ is playing $\black$. Nevertheless, a player having a \emph{threshold} equal to $0$ will always play $\black$ regardless of the actions played by his neighbors. Although we can restrict the thresholds to being non-zero, we refrain from doing so. Instead, we generalize our model (allowing $\black$ to be a best response for $i$ even if no player in $\N_i$ is playing $\black$) to provide symmetry between both actions $\black$ and $\white$. We do this for two reasons. The first is to study the linear threshold model as considered in the literature. The second is a technical reason, mainly to ensure closure of the set $\G_n{\times}\Q_n$ under certain operations e.g. node removal. Nevertheless, any result for the generalized version of the model is inherited by the initial version trivially by inclusion.

Having made the transition from \emph{types} to \emph{thresholds}, we distinguish the nodes having thresholds at the \emph{boundaries} as follows:
\begin{MyDef} \label{nonvalid}
Given a pair $(G,k)$ in $\G_n{\times}\K_n$, node $i$ in $\I_n$ is called non-valid with respect to $(G,k)$ (or simply non-valid) if $k_i$ is either equal to 0 or (strictly) greater than $d_i$. A node is called valid if it is not non-valid.
\end{MyDef}
A non-valid node is then allowed to play only one of the actions in $\{\white,\black\}$ whenever it is allowed to decide on the action to play. Node $i$ will always choose to play $\black$ if its threshold is $0$, and will always choose to play $\white$ if its threshold is (strictly) greater than $d_i$.

Finally, given a pair $(G,k)$ in $\G_n{\times}\K_n$, we denote by $G_k$ the map from $\A_n$ into $\A_n$ such that for player $i$, $(G_k a)_i = \black$ if and only if at least $k_i$ players are in $a^{-1}(\black)\cap \N_i$.\footnote{Let $f:A\rightarrow B$ and $g:B\rightarrow C$ be functions, we denote by $gf$ the function $g \circ f:A\rightarrow C$. In particular, if a function $f$ maps a set $A$ to itself, for a non-negative integer $m$, we denote by $f^m$ the function $f{\circ}f^{m-1}$ where $f^0$ is the identity map on $A$.} %
From this perspective, given an initial configuration $a$ in $\A_n$, the sequence $a,G_ka,G^2_ka,\cdots$ corresponds to the sequence of action profiles $a,a_1,a_2,\cdots$ where $a_T = G^T_ka$ is the action profile played by the players at time $T$ if they act in accordance with the rule in Proposition \ref{ruleInt}.


\subsection{The Limiting Behavior}
To understand the limiting behavior, we note two fundamental properties: the space $\A_n$ has finite cardinality, and Proposition \ref{ruleInt} is deterministic. Since $\A_n$ is finite, if we let $a_0,a_1,a_2,\cdots$ be any infinite sequence of action profiles played by the agents according to Proposition \ref{ruleInt}, then there exists at least one action profile $\hat{a}$ that will appear infinitely many times along this sequence. Since the dynamics are deterministic (and that actions at time T+1 depend only on actions at time T), the same sequence of action profiles appears between any two consecutive occurrences of $\hat{a}$. This means that after a finite time step, the sequence of action profiles will cycle among action profiles.

Let us consider a different representation of the dynamics. Given a pair $(G,k)$ in $\G_n{\times}\K_n$, we define a (binary) relation $\rightarrow$ on $\A_n$ such that for $a$ and $b$ in $\A_n$, $a \rightarrow b$ if and only if $b = G_ka$. The graph $H(\A_n,\rightarrow)$ then forms a directed graph (possibly with self loops) on the vertex set taken to be the space of action profiles $\A_n$, and an action profile $a$ is connected to an action profile $b$ by a directed edge $(a,b)$ going from $a$ to $b$ if and only if $b = G_ka$. Suppose we pick a vertex $a$, namely an action configuration, and perform a walk on vertices along the edges in $H$ starting from $a$. The walk eventually cycles vertices in the same order. Every initial action profile leads to one cycle, and two action profiles need not lead to the same cycle. We formalize the idea in the following definitions.

\begin{MyDef} \label{Rrelation}
Given $(G,k)$ in $\G_n\times\K_n$, for two action profiles $a$ and $b$ in $\A_n$, we say that $a$ can be reached from $b$ with respect to $G_k$ if there exists a non-negative integer $T$ such that $a = G_k^Tb$. Formally, we define the relation $\R_{G_k}$ on $\A_n$ such that for $a$ and $b$ in~$\A_n$, $a\R_{G_k}b$ if and only if there exists a non-negative integer $T$ such that $a = G_k^Tb$.
\end{MyDef}

For $a$ and $b$ in $\A_n$, we have $a\R_{G_k}b$ if and only if there exists a directed path in $H(\A_n,\rightarrow)$ from vertex $b$ to vertex $a$. The graph $H$ is not necessarily weakly-connected,\footnote{A directed graph $G$ is said to be weakly-connected if for any vertices $u$ and $v$ in the vertex set of $G$, there exists an \emph{undirected} path connecting $u$ to $v$. A weakly-connected component of $G$ is a maximal subgraph of $G$ that is weakly-connected.} and by the argument provided at the beginning of this subsection, every weakly-connected component of $H$ necessarily contains a directed cycle (possibly a self-loop). Moreover, each vertex in $H$ can have at most one outgoing edge, therefore every weakly-connected component of $H$ cannot contain more than one directed cycle.\footnote{If we construct a relation $\C$ on $\A_n$ such that for $a$ and $b$ in $\A_n$, $a\C b$ if and only if $a\R_{G_k}b \text{ or } b\R_{G_k}a$, then $\C$ is an equivalence relation on $\A_n$. In this setting, two configurations in $\A_n$ are in the same equivalence class with respect to the relation $\C$ if and only if they are in the same weakly-connected component in $H$.} We formally characterize the set of those cycles as follows:

\begin{MyDef} \label{CycleDef}
Given a pair $(G,k)$ in $\G_n{\times}\K_n$, we define $CYCLE_n(G,k)$ to be the collection of subsets of $\A_n$, such that for every $C$ in $CYCLE_n(G,k)$, if $a$ and $b$ are in $C$ then we have both $a\R_{G_k}b$ and $b\R_{G_k}a$, and for every $c$ in $\A_n\backslash C$, there does not exist an action configuration $a$ in $C$ such that $a\R_{G_k}c$. We refer to the elements of $CYCLE_n(G,k)$ as convergence cycles.
\end{MyDef}

The condition ``$a\R_{G_k}b \text{ and } b\R_{G_k}a$'' can be concisely replaced by ``$a\R_{G_k}b$'', however we keep it as such to stress on the fact that both $a$ can be reached from $b$ and $b$ can be reached from $a$. The second condition ensures that $C$ is in $CYCLE_n(G,k)$ only if there exists no larger cycle $C'$ containing $C$.

Cycles in $CYCLE_n(G,k)$ consisting of only one action configuration are fixed-points of $G_k$ and so will be referred to as \emph{fixed-points}. Cycles in $CYCLE_n(G,k)$ consisting of more than one action configuration will be referred to as \emph{non-degenerate cycles} (as opposed to \emph{fixed-points} which are \emph{degenerate cycles}).\\

In the remaining part of this section, we provide a broad overview of the main results (without proofs). We only mention results concerning the dynamical properties of the model (e.g. the length of convergence cycles, convergence time, the number of fixed points or cycles and reachability). We do not consider yet any resilience measure or bounds thereof, we do so in Section 7.

\subsection{An Overview of Convergence Results}

Given the limiting cyclic behavior, the most natural starting point would be to characterize the length of the cycles in the equivalence classes as a function of the imposed graph structure and the threshold distribution.
\begin{MyThm} \label{ConvCycl}
For every positive integer $n$, every $(G,k)$ in $\G_n{\times}\K_n$ and every cycle $C$ in $CYCLE_n(G,k)$, the cardinality of $C$ is less than or equal to $2$.
\end{MyThm}

Put differently, given a network structure $G$, a threshold distribution $k$ and an initial action profile $a$, if we iteratively apply $G_k$ on $a$ ad infinitum to get a sequence of best response action profiles, along the sequence of actions considered by player $i$, player $i$ will eventually either settle on playing one action, or switch action on every new application of~$G_k$.\\

We further show (in Sections 4.4 and 4.5) that such a limiting behavior also holds for the extension model. We then proceed to characterize the number of iterations needed to reach a convergence cycle.

\begin{MyDef}
For every positive integer $n$, and every $(G,k,a)$ in $\G_n{\times}\K_n{\times}\A_n$, we define $\delta_n(G,k,a)$ to be equal to the smallest non-negative integer $T$ such that there exists a cycle $C$ in $CYCLE_n(G,k)$ and $b$ in $C$ with $G_k^Ta = b$.
\end{MyDef}

The quantity $\delta_n(G,k,a)$ denotes to the minimal number of iterations needed until a given action configuration $a$ reaches a cycle, when iteratively applying $G_k$. We refer to $\delta_n(G,k,a)$ as the \emph{convergence time} from $a$ under $G_k$.

\begin{MyThm} \label{ConvTimeQuad}
For some positive integer $c$, every positive integer $n$, and every $(G,k,a)$ in $\G_n{\times}\K_n{\times}\A_n$, the convergence time $\delta_n(G,k,a)$ is less than or equal to $cn^2$.
\end{MyThm}

We further improve the results on convergence time and get a bound that is linear in the size of the network when the graphs are restricted to cycle graphs or trees.

\begin{MyThm} \label{ConvTimeLin}
For all positive integers $n$, and every $(G,k,a)$ in $\G_n{\times}\K_n{\times}\A_n$ where $G$ is an even-length cycle graph or a tree, the convergence time $\delta_n(G,k,a)$ is less than or equal to $n$.
\end{MyThm}

A linear bound can also be derived for complete graphs, but we do not consider this case in this paper. Instead, we refer the reader to \cite{MyThesis}.\\

We proceed to characterize the number of fixed-points and (non-degenerate) convergence cycles and present an overview of results on decision and counting problem that arise within this framework.

\subsection{An Overview of Complexity Results}

We first argue in Section 6 that no insightful \emph{uniform} bound on the number of limiting configurations exists; the range of such a bound will be too large. In this case, how well can we characterize the number of fixed-points and non-degenerate convergence cycles? To this end, we study the complexity of counting those numbers. We consider the counting problems\footnote{We refer the reader to Appendix \ref{ComplexityAppendix} for a short review of the required background.} \#CYCLE, \#FIX and \#2CYCLE that take $<n,G,k>$ as input, where $n$ is a positive integer and $(G,k)$ belongs to $\G_n\times\K_n$, and outputs the cardinality of $CYCLE_n(G,k)$, the number of fixed points and the number of non-degenerate cycles respectively.
\begin{MyThm} 
 \#CYCLE, \#FIX and \#2CYCLE are \#P-Complete.
\end{MyThm}
One has to be subtle towards what such result entails. This result does not imply that no characterization of the number of fixed-points or non-degenerate cycles is possible whatsoever, but rather that we would be unable to get an arbitrarily refined characterization of that number assuming \#P is not in FP.\\

We further show in Section 6.1 that those counting problems remain hard even if we restrict the graphs to be bipartite and impose homogeneous thresholds on the players.\\

We then proceed to the question of reachability whereby given a graph structure $G$, a type distribution $k$ we decide whether a certain action configuration can be reached from some other configuration. We define the language PRED to consist of all 4-tuples $<n,G,k,a>$, where $n$ is a positive integer, $(G,k,a)$ belongs to $\G_n\times\K_n\times\A_n$ with $G_k(a)^{-1} \neq \emptyset$.
\begin{MyThm}
 PRED is NP-Complete.
\end{MyThm}

Given a graph structure $G$, a type distribution $k$ and a configuration $a$, computing the number of configurations $b$ from which $a$ can be reached by applying $G_k$ only once on $b$ is then necessarily hard. Instead, we define the counting problem \#reachable-PRED to take $<n,G,k,a>$ as input, where $n$ is a positive integer, $(G,k,a)$ and $G_k(a)^{-1} \neq \emptyset$ and output the cardinality of $G_k^{-1}(a)$. We get the following result:
\begin{MyThm} 
 \#reachable-PRED  is \#P-Complete.
\end{MyThm}

The results are derived from thresholds in $\K_n$ instead of types of $\Q_n$. However, the results trivially extend to types as follows: Convergence results hold by inclusion; complexity results hold since they still hold if we restrict $(G,k)$ to contain no non-valid node.  We devote Sections 4 and 5 to convergence results, and Section 6 to complexity results.

\section{On Convergence Cycles}

We begin by studying the following problem: given a graph $G$ in $\G_n$ and a threshold distribution $k$ in $\K_n$, how many action configurations does a cycle in $CYCLE_n(G,k)$ contain? Ultimately, we show that for any graph and any threshold distribution, the cycles in $CYCLE_n(G,k)$ consist of at most two action configurations. We begin the analysis by considering cycle graphs\footnote{A cycle graph is a 2-regular connected graph, we use both terms interchangeably.} , then proceed to trees. We refer the reader to \cite{MyThesis} for an analysis on complete graphs. Each of those special cases is treated by exploiting its graphic properties. Obviously, most of those properties are not shared among all graphs, and some cannot even be generalized to general graphs. Nevertheless, we explicitly provide results over those toy examples to build up the intuition of the reader and construct the combinatorial framework slowly as we go along. After trees, we consider general graphs. We then generalize the results to the extension model.

\subsection{Cycle Graphs}\label{CRing}

Let us consider a pair $(G,k)$ in $\G_n{\times}\K_n$ where $G$ is a path, i.e. every agent is connected to at most two other agents and no cycles in the graphs are allowed. Our first intent is to characterize the length of the limiting cycles in that case. Suppose $k$ is picked in such a way that some players are non-valid\footnote{See definition \ref{nonvalid}}, then we know that those players can only play one action after some finite time step. With respect to the analysis concerned, we may remove those players, update the thresholds of the neighboring players accordingly and end up with a collection of disconnected paths. Restricting the analysis to one of the paths leads us back to the initial case. Therefore, we will assume that every node in the graph is valid: this implies that $k$ is equal to $1$ for the nodes having degree $1$ and $k$ takes values in $\{1,2\}$ for the nodes having degree equal to 2. Moreover, to take care of the boundary case, we will connect the 1-degree nodes together, and so forming a ring of agents. The graph in consideration is then the 2-regular connected graph. We then relax $k$ to take values in $\{1,2\}$ over $\I_n$.\\
 
Given that the thresholds of the nodes are either $1$ or $2$, it is useful to explicitly state the decision rules as follows. Let $a$ be some action configuration in $\A_n$, if node $i$ has a threshold $k_i$ equal to $1$, then node $i$ is $\black$ in $G_ka$ if and only if either one of its neighbors is $\black$ in $a$. Similarly, if node $i$ has a threshold $k_i$ equal to $2$, then node $i$ is $\black$ in $G_ka$ if and only if both of its neighbors are $\black$ in $a$.

Let us impose a strict ordering on $\{\white,\black\}$ such that $\min\{\white,\black\} = \white$. This translates to $\white \wedge \black = \white$ for notational convenience.\footnote{ \label{veewedge} For $x_1$ and $x_2$ in a strictly ordered set, we denote $\max\{x_1,x_2\}$ and $\min\{x_1,x_2\}$ by $x_1 \vee x_2$ and $x_1 \wedge x_2$ respectively.} For elements $a$, $b$ and $c$ in $\{\white,\black\}$, the following identities can be checked:
\begin{align}
a \wedge a = a																					&\qquad a \vee a = a\nonumber\\
a \wedge \black = a																			&\qquad a \vee \white = a\nonumber\\
a \wedge b = b \wedge a 																&\qquad  a \vee b = b \vee a\nonumber\\
a \wedge (b \vee c) = ( a \wedge b) \vee (a \wedge c) 	&\qquad a \vee (b \wedge c) = ( a \vee b) \wedge (a \vee c)\nonumber\\
a \wedge (b \wedge c) = (a \wedge b) \wedge c 					&\qquad a \vee (b \vee c) = (a \vee b) \vee c.\nonumber
\end{align}

Given a pair $(G,k)$ in $\G_n {\times} \K_n$ where $G$ is 2-regular and $k$ takes values in $\{1,2\}$, we define a map $\tau$ from $\I_n$ into $\{\vee, \wedge\}$ such that $\tau_i = \vee$ if and only if $k_i = 1$.
Similarly, let us define two maps $s$ and $p$ from $\I_n$ into $\I_n$ (we refer to them as successor and predecessor) such that 
$i$ and $s_i$ are neighbors, $i$ and $p_i$ are neighbors and $(sp)_i = (ps)_i = i$. The dynamics are then represented as follows:
\begin{equation}
 (G_k a)_i =   a_{p_i} \tau_i a_{s_i}. \nonumber
\end{equation}

We give a quick example to illustrate. Let us consider a 2-regular connected graph $G$ over the set $\I_n$ for $n \geq 5$ and suppose we are given a threshold distribution $k$ in $\K_n$ taking values in $\{1,2\}$. Let us choose a node $i$ from $\I_n$. The nodes $s_i$ and $p_i$ are then (distinct) neighbors of node $i$. Let $a$ be an action configuration in $\A_n$ and suppose $k_i = 1$,  then $(G_k a)_i = \black$ if and only if either $a_{s_i} = \black$ or $a_{p_i} = \black$ i.e. at least one neighbor is $\black$. We can rewrite the previous statement as:
\begin{equation}
(G_k a)_i = \max\{a_{s_i},a_{p_i}\} = a_{s_i} \vee a_{p_i} = a_{s_i} \tau_i a_{p_i}. \nonumber
\end{equation}
This follows from the fact that we imposed a strict ordering on $\{\white,\black\}$ such that $\min\{\white,\black\} = \white$, and defined $\tau_i = \vee$ if and only if $k_i = 1$.
Furthermore, $s_i$ has both $i$ and $(ss)_i$ as neighbors. Suppose $k_{s_i} = 2$, then $(G_k a)_{s_i} = \black$ if and only if both $a_{i} = \black$ and $a_{(ss)_i} = \black$ i.e. at least two neighbor are $\black$. Similarly, we can rewrite the previous statement as:
\begin{equation}
(G_k a)_{s_i} = \min\{a_{i},a_{(ss)_i}\} = a_{i} \wedge a_{(ss)_i} = a_{i} \tau_{s_i} a_{(ss)_i}. \nonumber
\end{equation}
To conclude the example, we further suppose that $k_{p_i} = 2$. The node $p_i$ has both $i$ and $(pp)_i$ as neighbors, and similarly to the rule of node $s_i$, we have $(G_k a)_{p_i} = a_{i} \wedge a_{(pp)_i}$. We may now express $(G^2_k a)_i$ in terms of actions in $a$ as follows:
\begin{equation}
 (G^2_k a)_i = (G_k (G_ka))_i = (G_ka)_{s_i} \vee (G_ka)_{p_i} = (a_{i} \wedge a_{(ss)_i}) \vee (a_{i} \wedge a_{(pp)_i}). \nonumber
\end{equation}
Using distributivity (as defined in the identities earlier) we get:
\begin{equation}
	(a_{i} \wedge a_{(ss)_i}) \vee (a_{i} \wedge a_{(pp)_i}) = a_{i} \wedge (a_{(ss)_i} \vee  a_{(pp)_i}). \nonumber
\end{equation}

We now generalize the last part of the example. We consider a pair $(G,k)$ in $\G_n {\times} \K_n$ where $G$ is a cycle graph and $k$ takes values in $\{1,2\}$. In this setting, $(G_k^2 a)_i$ would only depend on the actions of nodes $(ss)_i$, $(pp)_i$ and $i$ itself in $a$. In particular, there are a total of eight possible decision rules, we summarize them in the following table:\\
\begin{center}
	\begin{tabular}{l ccc c  l } 
	\qquad & $\tau_{p_i}$ & $\tau_{i}$ & $\tau_{s_i}$ & $\qquad (G_k a)_i \qquad$ & $\qquad (G_k^2 a)_i \quad $\\[0.5cm]
	{c.1} \qquad& $\vee$ & $\vee$ & $\vee$ & $a_{p_i} \vee a_{s_i}$ &  $a_i \vee (a_{(ss)_i} \vee a_{(pp)_i})$\\
	{c.2} \qquad& $\vee$ & $\vee$ & $\wedge$ & $a_{p_i} \vee a_{s_i}$ & $a_i  \vee a_{(pp)_i}$\\
	{c.3} \qquad& $\vee$ &  $\wedge$ & $\vee$ & $a_{p_i} \wedge a_{s_i}$ & $a_i \vee (a_{(ss)_i} \wedge a_{(pp)_i})$\\
	{c.4} \qquad& $\vee$ & $\wedge$ & $\wedge$ & $a_{p_i} \wedge a_{s_i}$ & $a_i \wedge a_{(ss)_i}$\\
	{c.5} \qquad& $\wedge$ & $\vee$ & $\vee$ & $a_{p_i} \vee a_{s_i}$ & $a_i \vee a_{(ss)_i}$\\
	{c.6} \qquad& $\wedge$ & $\vee$ & $\wedge$ & $a_{p_i} \vee a_{s_i}$ & $a_i \wedge (a_{(ss)_i} \vee a_{(pp)_i})$\\
	{c.7} \qquad& $\wedge$ & $\wedge$ & $\vee$ & $a_{p_i} \wedge a_{s_i}$ & $a_i \wedge a_{(pp)_i}$\\
	{c.8} \qquad& $\wedge$ & $\wedge$ & $\wedge$ & $a_{p_i} \wedge a_{s_i}$ &  $a_i \wedge (a_{(ss)_i} \wedge a_{(pp)_i})$\\[0.5cm]
	\end{tabular}
\end{center}

We proceed by defining strong assignments, then state a first proposition on the dynamics induced by the above table. 

\begin{MyDef} \label{StrongAssignment}
For any positive integer $n$ greater than 1, any graph $G$ in $\G_n$, every threshold distribution $k$ in $\K_n$, an action $c$ is called a strong action (or strong assignment) for player $i$ in $\I_n$ if once played by player $i$ at time step $T$, it is played by player $i$ at time step $T+2m$ for all positive integers $m$, regardless of what is played by the neighbors of node $i$.
\end{MyDef}

Given the definition, every node in the setting concerned in this subsection has a strong assignment.

\begin{MyPro} \label{RingSA}
For any integer $n$ greater than 2, any 2-regular graph $G$ in $\G_n$, any threshold distribution $k$ in $\K_n$ taking values in $\{1,2\}$, every player $i$ in $\I_n$ has a strong assignment. 
\end{MyPro}

\begin{proof}
Let $i$ be a player in $\I_n$, to prove the result it would be enough to investigate the update rule over two iterated applications of $G_k$ of player $i$, i.e. the value $(G_k^2 a)_i$ takes. Referring back to the previous table, notice that in the case of each threshold distribution over $\{p_i,i,s_i\}$, for $(G_k^2 a)_i$ to be equal to $\black$, the condition $a_i = \black$ is either sufficient or necessary. In the case where $a_i = \black$ is sufficient, if $a_i=\black$ then $(G_k^{2m} a)_i = \black$ for all positive integers $m$, and so $\black$ is a strong assignment for player $i$. Likewise, in the case where $a_i = \black$ is necessary, if $a_i=\white$ then $(G_k^{2m} a)_i = \white$ for all positive integers $m$, and so $\white$ is a strong assignment for player $i$.   
\end{proof}

In particular, c.1, c.2, c.3 and c.5 correspond to $\black$ being a strong assignment for player $i$, and c.4, c.6, c.7 and c.8 correspond to $\white$ being a strong assignment for player $i$. We now characterize the length of the convergence cycles.

\begin{MyPro} \label{RingConv}
For any integer $n$ greater than 2, any 2-regular graph $G$ in $\G_n$ and any threshold distribution $k$ in $\K_n$ taking values in $\{1,2\}$, each cycle $C$ in $CYCLE_n(G,k)$ has cardinality less than or equal to $2$.
\end{MyPro}

\begin{proof}
Let $a$ be an action configuration in $\A_n$ and suppose we construct the sequence $a,G_ka, G^2_ka, G^3_ka,\cdots$. For notational convenience, let us denote $G^T_ka$ by $a^T$. We consider the subsequence $a^0,a^2,a^4,\cdots$, choose a player $i$ in $\I_n$ and then observe the evolution of the action played by player $i$ over two time step, i.e. we consider the sequence $a^0_i,a^2_i,a^4_i,\cdots$. Without any loss of generality, let us assume that $\black$ is the strong assignment. Either $\black$ appears in the sequence or $\black$ does not appear in the sequence. If it does appear, then there exists a positive integer $M$ such that $a^{2m}_i = \black$ for all $m \geq M$. If it does not appear, then $a^{2m}_i = \white$ for all non-negative integers $m$. Either way, for every player $i$, there exists a non-negative integer $T_i$ and an action $c$ in $\{\white,\black\}$ such that $a^{2m}_i = c$ for all $m\geq T_i$. If we set $T = \max_i T_i$, then there exists an action profile $\hat{a}$ such that $a^{2m} = \hat{a}$ for all $m\geq T$. We then get that $\{\hat{a},G_k\hat{a}\}$ is the cycle reached from $a$. It follows that if we let $C$ be any cycle in $CYCLE_n(G,k)$ and we let $a$ be an action configuration in $C$, then necessarily $C = \{a,G_ka\}$.  
\end{proof}

We transition to investigate the behavior when the graph structure is a tree.

\subsection{Trees}

We consider in this section dynamics on trees, namely acyclic connected graphs. In this section, the letter $T$ shall always be used to denote trees, and never time as was done sometimes in previous sections. Given a tree $T$ in $\G_n$, if we label a node $r$ in $\I_n$ as \emph{root}, the children of node $i$ (with respect to the root $r$) are all the neighbors of $i$, not lying on the path from the root~$r$ to node~$i$.
Finally, a leaf in the tree $T$ is a node having degree $1$. Fortunately, strong assignments appear in the dynamics on trees. We begin by stating the following proposition:

\begin{MyPro} \label{TreeSA}
For any integer $n$ greater than 1, any tree $T$ in $\G_n$ and any threshold distribution $k$ in $\K_n$ such that all nodes are valid with respect to $(G,k)$, pick a root $r$ for the tree, then for every node $i$ where all its children (with respect to $r$) are leaves, $i$ has at least one strong assignment. In particular, if $k_i>1$, then $\white$ is a strong assignment and if $k_i<d_i$ then $\black$ is a strong assignment.
\end{MyPro}

\begin{proof}
We know that for each node $i$ where all its children (with respect to $r$) are leaves, node $i$ has at least $d_i - 1$ leaves connected to it. Since all nodes are considered to be valid, then each leaf has a threshold of $1$. Suppose $k_i > 1$, it then follows that $(G_k^2 a)_i = a_i \wedge \phi(a)$ for some map $\phi$ from $\A_n$ into $\{\white,\black\}$. Suppose $k_i < d_i$, it then follows that $(G_k^2 a)_i = a_i \vee \phi(a)$ for some map $\phi$ from $\A_n$ into $\{\white,\black\}$.
\end{proof}

In this case, note that if $1<k_i<d_i$, then $i$ has both $\black$ and $\white$ as strong assignment. This fact implies that $i$ will never change its color over two time steps. We note that the proposition considers only the case where all nodes in concern are valid.\\

Aside being acyclic, trees enjoy bipartiteness: a crucial property that will be heavily relied on when considering general graphs. We begin to convey how the bipartite property of graphs may be exploited. The definitions and results to follow apply to general bipartite graphs.

\begin{MyDef}
Let $P$ be a subset of $\I_n$, for $(G,k)$ in $\G_n{\times}\K_n$, we define $G_k|_P$ to be the restriction of $G_k$ to act on the actions of the players in $P$. Formally, for $a$ in $\A_n$,
\begin{equation}
(G_k|_Pa)_i =
\left\{
	\begin{array}{cl}
		(G_ka)_i & \text{if } i\in P \\
		a_i & \text{if } i\notin P
	\end{array} \nonumber
\right.
\end{equation}
\end{MyDef}
We note that we are not restricting the domain of the function, $G_k|_P$ is indeed a map from $\A_n$ into $\A_n$. To proceed, it is known that any bipartite graph has a 2-(node)-coloring. We avoid the wording \emph{coloring} to avoid confusion. Instead, we define \emph{2-Partitions}. Let $\G^b_n$ be the set of all connected undirected bipartite graphs defined over the vertex set $\I_n$.

\begin{MyDef}
Given a graph $G(\I_n,E^b)$ in $\G^b_n$, a 2-Partition of $\I_n$ with respect to $G$, is a pair $(P_o,P_e)$ of disjoint subsets of $\I_n$ such that $P_o \cup P_e = \I_n$ and there does not exist an $(i,j)$ in $P^2_o \cup P^2_e$ such that $ij \in E^b$.
\end{MyDef}

We eventually restrict $G_k$ to act on the nodes in $P_o$ and $P_e$ separately. For convention, $o$ would refer to \emph{odd} and $e$ to \emph{even}. The dynamics will be presented in such a way, that nodes in $P_o$ (resp. $P_e$) will be allowed to change actions only at odd (resp. even) time steps. Let us first clearly define a partition of a set.

\begin{MyDef}
Let $X$ be a set, a partition $P_1,\cdots,P_m$ of $X$ is a finite collection of disjoint non-empty subsets of $X$ whose union is $X$.
\end{MyDef}
The definition to follow serves mainly as a notational clarification, its technical value is rather intuitive.

\begin{MyDef}
 Consider a function $f$ mapping $\I_n$ into some set. Let $P_1,\cdots,P_m$ be a partition of $\I_n$, and let $f{\upharpoonright}P_l$ be the restriction of $f$ to have domain $P_l$. Let $\pi$ be any permutation on $\{1,\cdots,m\}$, we consider $f$ to be equal to $(f{\upharpoonright}P_{\pi(1)},\cdots, f{\upharpoonright}P_{\pi(m)})$. 
\end{MyDef}

Given a 2-Partition, we may `decouple' the dynamics and the following identities would emerge:

\begin{MyPro} \label{Identities}
Given a pair $(G,k)$ in $\G^b_n {\times} \K_n$, if we consider a 2-Partition $(P_o,P_e)$ of $\I_n$ with respect to $G$, then:
\begin{enumerate}
 \item $G_ka = ((G_k|_{P_o}a){\upharpoonright}P_o, (G_k|_{P_e}a){\upharpoonright}P_e )$
 \item $G_k|_{P_e} G_k|_{P_o} a = (G_ka{\upharpoonright}P_o, G^2_ka{\upharpoonright}P_e )$
 \item $G^2_ka = ( (G_k|_{P_e} G_k|_{P_o} a){\upharpoonright}P_e, (G_k|_{P_o} G_k|_{P_e} a){\upharpoonright}P_o)$.
\end{enumerate}
\end{MyPro}

\begin{proof}
The fact that $G_ka = ((G_k|_{P_o}a){\upharpoonright}P_o, (G_k|_{P_e}a){\upharpoonright}P_e )$ follows from the definition of $G_k$.
We have $(G_k|_{P_o}a){\upharpoonright}P_o = (G_ka) {\upharpoonright}P_o$, so $(G_k|_{P_e} G_k|_{P_o} a) {\upharpoonright}P_o = (G_ka) {\upharpoonright}P_o$, and on the other hand since $G_k|_{P_e}$ modifies the action of the players in $P_e$ based only on the actions in $P_o$, we get: 
\begin{align}
(G_k|_{P_e} G_k|_{P_o} a) {\upharpoonright}P_e &= G_k|_{P_e} ((G_ka) {\upharpoonright}P_o,a {\upharpoonright}P_e){\upharpoonright}P_e\nonumber\\
											 &= G_k|_{P_e}((G_ka) {\upharpoonright}P_o, (G_ka) {\upharpoonright}P_e){\upharpoonright}P_e\nonumber\\
											 &= G_k|_{P_e}G_k a {\upharpoonright}P_e\nonumber\\
											 &= G^2_ka{\upharpoonright}P_e.\nonumber
\end{align}
As for the last statement:
\begin{align}
	G^2_ka{\upharpoonright}P_o &= G_k(G_ka{\upharpoonright}P_e, G_ka {\upharpoonright}P_o){\upharpoonright}P_o\nonumber\\
							&= G_k|_{P_o} (G_ka{\upharpoonright}P_e, G_ka {\upharpoonright}P_o){\upharpoonright}P_o\nonumber\\
						    &= G_k|_{P_o} (G_k|_{P_e}a{\upharpoonright}P_e, G_k|_{P_e}a {\upharpoonright}P_o){\upharpoonright}P_o\nonumber\\
						    &= (G_k|_{P_o}G_k|_{P_e}a){\upharpoonright}P_o.\nonumber
\end{align}
Similarly, we get $G^2_ka{\upharpoonright}P_e = (G_k|_{P_e}G_k|_{P_o}a){\upharpoonright}P_o$.
\end{proof}

This fact allows us to say something about the sequence $a, G_ka, G^2_k a, \cdots$ for all $a$ by observing  $b, G_k|_{P_e} G_k|_{P_o} b , (G_k|_{P_e} G_k|_{P_o})^2 b, \cdots$ for all $b$. It allows us to observe the process \emph{diagonally}, in a zig-zag fashion. We shall use this fact, and we state the following lemma to formalize the idea.

\begin{MyPro} \label{TreeSplit}
Given a pair $(G,k)$ in $G^b_n {\times} \K_n$, for every cycle $C$ in $CYCLE_n(G,k)$, if for some action configuration $a$ in $C$, for any 2-Partition $(P_o,P_e)$ of $\I_n$ with respect to $G$ we have $(G_k|_{P_e} G_k|_{P_o} a){\upharpoonright}P_e = a{\upharpoonright}P_e$ then $C$ has a cardinality of at most $2$.
\end{MyPro}

Of course, the condition could have also been written as $(G_k|_{P_o} G_k|_{P_e} a){\upharpoonright}P_o = a{\upharpoonright}P_o$, but that would make no difference in the previous statement: we need the property to hold for all 2-Partitions.

\begin{proof}
From Proposition \ref{Identities}, we have $G^2_ka = ( (G_k|_{P_e} G_k|_{P_o} a){\upharpoonright}P_e, (G_k|_{P_o} G_k|_{P_e} a){\upharpoonright}P_o)$. Let $C$ be a cycle in $CYCLE_n(G,k)$, if we have $(G_k|_{P_e} G_k|_{P_o} a){\upharpoonright}P_e = a{\upharpoonright}P_e$ and $(G_k|_{P_o} G_k|_{P_e} a){\upharpoonright}P_o = a{\upharpoonright}P_o$ for some action configuration in $C$, then $G_k^2a = a$ and so $C$ has a cardinality of at most $2$.
\end{proof}

To prove that all cycle have cardinality at most $2$, we will study for all 2-Partitions $(P_o,P_e)$ and every $a$ in $\A_n$, the sequence $a, G_k|_{P_e} G_k|_{P_o} a , (G_k|_{P_e} G_k|_{P_o})^2 a, \cdots$ and show that this sequence is eventually constant, i.e. there exists a finite time step after which all terms in the sequence become equal. If that is the case, then for any $a$ in $\A_n$, the sequence $a, G_k^2a, G_k^4a, \cdots$ is eventually constant and so cycles cannot have a cardinality greater than 2. To this end, we note the following fact:

\begin{MyPro} \label{TreeCond}
Given a pair $(G,k)$ in $\G^b_n{\times}\K_n$ and a 2-Partition $(P_o,P_e)$ of $\I_n$, for every player $i$ in $\I_n$, if $c \in \{\white,\black\}$ is a strong assignment for player $i$, then for every $a$ in $\A_n$, if $(G_k|_{P_e} G_k|_{P_o} a)_i = c$, then $((G_k|_{P_e} G_k|_{P_o})^m a)_i = c$ for all positive integers $m$.
\end{MyPro}

\begin{proof}
 If $i$ belongs to $P_e$, we have $(G_k|_{P_e} G_k|_{P_o} a)_i = (G^2_k a)_i$, and the result then follows from the definition of strong assignment. If $i$ belongs to $P_o$, then $(G_k|_{P_e} G_k|_{P_o} a)_i = (G_k|_{P_o} a)_i$. So if $(G_k|_{P_e} G_k|_{P_o} a)_i = c$ then $(G_k|_{P_o} a)_i = c$ and since:
 \begin{align}
 	((G_k|_{P_e} G_k|_{P_o})^m a)_i &= (G_k|_{P_e} (G_k|_{P_o}G_k|_{P_e})^{m-1}G_k|_{P_o}  a)_i\nonumber\\
 																	&= ((G_k|_{P_o}G_k|_{P_e})^{m-1}G_k|_{P_o}  a)_i\nonumber\\
 																	&= (G^{2m-2}_k G_k|_{P_o}  a)_i,\nonumber
 \end{align}
the result would follow from the definition of strong assignment.
\end{proof}

We present a rather natural statement that will be used in the inductive argument to prove the theorem. Although the statement is intuitively simple, we give an effort to state it carefully so to take care of any minor technicalities involved. We first give a simplistic informal version of it for general graphs: given a triplet $(G,k,a)$ in $\G_n {\times}\K_n {\times} \A_n$, let us consider the sequence $a,G_ka,G_ka^2,\cdots$. If some player \emph{never} changes action along this sequence, then we may delete the player from the graph and modify the thresholds of the neighboring players in such a way that the effect of the action played by the player is `seen' by the neighbors. We do so by keeping the thresholds of the neighbors unchanged if that action is $\white$, and by decreasing the thresholds by $1$ (to a minimum of zero) if that action is $\black$. In this case, we would obtain a different triplet $(G',k',a')$ where all the players in $V(G')$ play the exact actions as in the sequence $a,G_ka,G_ka^2,\cdots$. We now specialize the previous idea to meet our needs for the bipartite case.

\begin{MyPro}\label{BE}
Let $(G,k)$ be a pair in $\G^b_n{\times}\K_n$, $(P_o,P_e)$ be a 2-Partition of $\I_n$ and $a$ an action configuration in $\A_n$. If there exists a player $i$ and an action $c$, such that $((G_k|_{P_e} G_k|_{P_o})^m a)_i = c$ for all non-negative integers $m$, consider $H$ to be the induced subgraph of $G$ over $\I_n\backslash\{i\}$. Suppose $G'$ is a connected component of $H$ with vertex set $\J$, define $a'$ to be the action configuration $a$ restricted to the players in $\J$, define $P'_o$ and $P'_e$ to be $P_o \cap \J$ and $P_e \cap \J$ respectively, and define $k'$ to be the map from $\J$ into $\Zgeqz$ such that $k' = k$ on $\J\backslash \N_i$, $k' = k$ on $\N_i \cap \J$ if $c = \white$ and $k' = (k-1) \vee 0$ on $\N_i \cap \J$ if $c = \black$. Then:
\begin{equation}
 ((G_k|_{P_e} G_k|_{P_o})^m a)_j = ((G'_{k'}|_{P'_e} G'_{k'}|_{P'_o})^m a')_j,\nonumber
\end{equation}
for all non-negative integers $m$ and all players $j$ in $\J$.
\end{MyPro}

\begin{proof}
It would be enough to show that the local decision rules of the players in $\J \cap \N_i$ does not change, simply that
\begin{equation}
 (G_k|_{P_e} G_k|_{P_o} a)_j = (G'_{k'}|_{P'_e} G'_{k'}|_{P'_o} a')_j.\nonumber
\end{equation}
Let $j$ be a player in $\J \cap \N_i$, either $j$ belongs to $P_o$ or $j$ belongs to $P_e$. Suppose $j$ belongs to $P_o$, then $(G_k|_{P_e} G_k|_{P_o} a)_j = (G_k|_{P_o} a)_j$, and $(G_k|_{P_o} a)_j = \black$ if and only if at least $k_j$ nodes in $\N_j$ play $\black$ in $a$, or equivalently at least $k'_j$ nodes in $\N_j \backslash \{i\}$ play $\black$ in $a'$, since $k'$ takes into account the action of player $i$. And that is equivalent to $(G'_{k'}|_{P'_e} G'_{k'}|_{P'_o} a')_j = \black$. Similarly, if $j$ belongs to $P_e$ we get that $(G_k|_{P_e} G_k|_{P_o} a)_j = (G'_{k'}|_{P'_e} (G_{k}|_{P_o} a {\upharpoonright} \J))_j = (G'_{k'}|_{P'_e}G'_{k'}|_{P'_o} a)_j$.
\end{proof}

With this settled, we go now to proving our theorem.

\begin{MyPro} \label{TreeConv}
For any positive integer $n$, any tree $T$ in $\G_n$ and any threshold distribution $k$ in $\K_n$, every cycle $C$ in $CYCLE_n(T,k)$ has cardinality less than or equal to $2$.
\end{MyPro}

\begin{proof}
 We prove this statement by induction. The statement can be exhaustively checked for $n=2$, and it is trivial for $n=1$. So let us assume that it holds when the number of players is less than or equal to $n\geq 3$. Let $(T,k,a)$ be a triplet in $\G^b_{n+1}{\times}\K_{n+1}{\times}\A_{n+1}$ where $T$ is a tree and consider a 2-Partition $(P_o,P_e)$ of $\I_{n+1}$. We consider the sequence $a, G_k|_{P_e} G_k|_{P_o} a , (G_k|_{P_e} G_k|_{P_o})^2 a, \cdots$. We claim that at least one node will never change colors along the sequence after some finite time step. If there exists a non-valid node, then clearly that is the case. If all nodes are valid, then (since $n\geq 2$) necessarily there exists a parent and a leaf. By Lemma \ref{TreeSA}, the parent has a strong assignment, and so it stops changing colors along the sequence after some finite time step. Either way, we may then apply Proposition \ref{BE} on that node at the time step when the color becomes constant, removing the node and updating the thresholds accordingly. The result follows since it holds on the obtained connected components by assumption.
\end{proof}

We proceed to extend such results to general graphs. Unfortunately, the notion of strong assignment does not extend to the general case, we construct a different approach.

\subsection{General Graphs}

Given a pair $(G,k)$ in $\G_n{\times}\K_n$ we will have this pair undergo two procedures: a bipartite-expansion and a symmetric-expansion. The bipartite-expansion will produce a pair $(G',k')$ where $G'$ is bipartite, and the symmetric-expansion will produce a pair $(G'',k'')$ where $G''$ has only nodes with odd degrees and $k''$ induces a majority decision rule. We show that by restricting the set of initial action profiles on $(G',k')$ and $(G'',k'')$, these two pairs simulate $(G,k)$ (see Lemmas \ref{OneStepM}, \ref{SymmetricM} and \ref{BipartiteM}). Given a pair $(G,k)$, we are then able to assume without any loss of generality (in proving the bound) that $G$ is bipartite, has only odd-degree nodes and that $k$ induces a majority decision rule. The bipartite property allows us to consider a 2-Partition of the graph and decouple the parallel decision scheme into a sequential process (with respect to the 2-Partition) as was done in the case of Trees. The majority decision rule then allows us to prove (by counting specific edges in the graph) that no node can flip infinitely many times along this sequential process. This argument establishes fixed-point convergence of the sequential process, translating to a upper-bound of $2$ on the length of the convergence cycles for the parallel dynamics.\\

For technical and notational convenience, we will assume (throughout the rest of the section) that when dealing with only two graphs $G$ and $G'$ in $\G_n$ and $\G_{n'}$ respectively with $n < n'$, the set $V(G)$ lies inside $V(G')$, i.e. $\I_n$ lies inside $\I_n'$. This only removes the need to explicitly identify a subset of $V(G')$ to $V(G)$.\\

First, recall that we denote by $\G^b_n$ the set of all connected undirected bipartite graphs defined over the vertex set $\I_n$. We begin by formally introducing the bipartite-expansion and the symmetric-expansion.

\begin{MyDef} \label{Bexp}
Given a pair $(G,k)$ in $\G_n {\times} \K_n$ where $G$ is non-bipartite, we construct a pair $(G',k')$ in $\G^b_{2n} {\times} \K_{2n}$ as outlined in the procedure to follow. We refer to $(G',k')$ as the bipartite-expansion of $(G,k)$.
\end{MyDef}

{\bf Bipartite-expansion procedure} --- Given a pair $(G,k)$ in $\G_n {\times} \K_n$ where $G$ is non-bipartite and has edge set $E$, we construct a pair $(G',k')$ in $\G^b_{2n} {\times} \K_{2n}$ as follows. We suppose $G'$ is equal to $(\I_{2n},E')$, and partition $\I_{2n}$ into two sets $\I_n$ and $\J$. We define a bijection $\phi$ from $\J$ into $\I_n$ and define $E'$ to be the set of undirected edges on $\I_{2n}$ such that for $i$ and $j$ in $\I_{2n}$, $\{i,j\} \in E'$ if and only if $(i,j)\in \I_n {\times} \J$ and $\{i,\phi(j)\} \in E$. Finally, set $k'$ to be equal to $k$ on $\I_n$ and $k \circ \phi$ on $\J$.

\begin{MyExa}
Under a bipartite-expansion, a cycle over $3$ nodes (left) becomes a cycle over $6$ nodes (right).
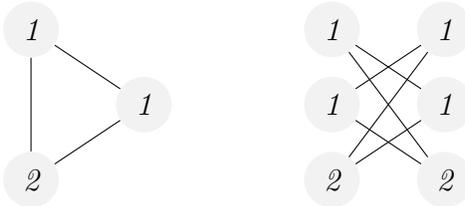
\begin{center}
\begin{tikzpicture}
  [scale=.5,auto=center,every node/.style={circle,fill=gray!10}]
  \node (n1) at (1,10) {1};
  \node (n2) at (4,8)  {1};
  \node (n3) at (1,6)  {2};

    \node (n4) at (9,10) {1};
  \node (n5) at (9,8)  {1};
  \node (n6) at (9,6)  {2};
     \node (n7) at (12,10) {1};
  \node (n8) at (12,8)  {1};
  \node (n9) at (12,6)  {2};
  
    \foreach \from/\to in {n1/n2,n2/n3,n1/n3,n4/n8,n4/n9,n5/n7,n5/n9,n6/n7,n6/n8}
    \draw (\from) -- (\to);
\end{tikzpicture}
\end{center}
The integer tag on each node represents its corresponding threshold.
\end{MyExa}

We define $\Sy_n$ to be the set of all pairs $(G,k)$ in $\G_n{\times}\K_n$ such for each player $i$ in $\I_n$, the degree $d_i$ is odd and $k_i$ is equal to $(d_i+1)/2$. We refer to $\Sy_n$ as the set of \emph{symmetric models}, in the sense that for $(G,k)$ in $\Sy_n$ the property is such that for any action profile $a$ in $\A_n$, and any player $i$, the action $(G_k a)_i$ is the action played by the majority in $\N_i$ with respect to the action profile $a$. In this case, the two actions $\black$ and $\white$ are treated as having equal weights by all players in the network.

\begin{MyDef} \label{OSexp}
Given a pair $(G,k)$ in $(\G_n \times \K_n) \backslash \Sy_n$, we construct a pair $(G',k')$ in $\G_n' \times \K_n'$ as outline in the procedure to follow. We refer to $(G',k')$ as a one-step symmetric-expansion of $(G,k)$.
\end{MyDef}

{\bf One-step symmetric-expansion procedure} --- Given a pair $(G,k)$ in $(\G_n \times \K_n) \backslash \Sy_n$, we construct a pair $(G',k')$ in $\G_n' \times \K_n'$ as follows. We suppose that $G$ is equal to $(\I_n,E)$, and choose a player $i$ in $\I_n$ such that either $d_i$ is even, or $d_i$ is odd and $k_i$ is not equal to $(d_i + 1)/2$. Surely such a node exists since $(G,k)$ does not belong to $\Sy_n$. We call the node $i$ the {\bf pivot} node in the one-step symmetric-expansion of $(G,k)$ into $(G',k')$. We construct an instance $(G',k')$ in $\G_{n+3(d_i + 1)}{\times}\K_{n+3(d_i + 1)}$. We suppose that $G'$ is equal to $(\I_{n+3(d_i + 1)},E')$ and partition $\I_{n+3(d_i + 1)}$ into $\I_n,P_1,\cdots,P_{d_i + 1}$ where each partition different than $\I_n$ has cardinality exactly equal $3$. We define $E'$ to be the undirected set of edges such that $E'$ contains $E$, and for every $m$, we suppose $P_m = \{j,j',j''\}$ and let $E'$ contain $jj'$, $jj''$ and $ij$. To visualize the obtained graph structure $G'$, we attached $d_i+1$ three-node Y-shaped graphs to node $i$. Finally, we set $k'$ to be equal to $k$ on $\I_n\backslash\{i\}$, to be equal to $d_i + 1$ at $i$, equal to $2$ on the remaining nodes having degree $3$ and equal to $1$ everywhere else.

\begin{MyExa}
The figure below shows a one-step symmetric expansion (right) of a pair $(G,k)$ (left) where $G$ is a cycle graph over 3 nodes.
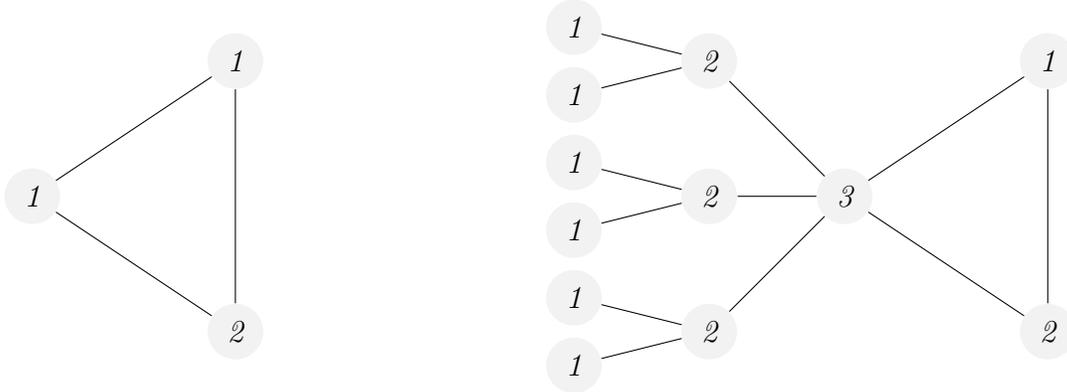
\begin{center}
\begin{tikzpicture}
  [scale=.9,auto=center,every node/.style={circle,fill=gray!10}]
  \node (n1) at (4,10) {1};
  \node (n2) at (1,8)  {1};
  \node (n3) at (4,6)  {2};

  \node (n4) at (9,10.5) {1};
  \node (n5) at (9,9.5)  {1};
  \node (n6) at (9,8.5)  {1};
  \node (n7) at (9,7.5) {1};
  \node (n8) at (9,6.5)  {1};
  \node (n9) at (9,5.5)  {1};
  
  \node (n10) at (11,10) {2};
  \node (n11) at (11,8)  {2};
  \node (n12) at (11,6)  {2};
  
  \node (n13) at (16,10) {1};
  \node (n14) at (13,8)  {3};
  \node (n15) at (16,6)  {2};

    \foreach \from/\to in {n1/n2,n2/n3,n1/n3,n10/n4,n10/n5,n10/n14,n11/n6,n11/n7,n11/n14,n12/n8,n12/n9,n12/n14,n14/n13,n14/n15,n13/n15}
    \draw (\from) -- (\to);
\end{tikzpicture}
\end{center}
The integer tag on each node represents its corresponding threshold. The pivot in this example is a node (in the cycle graph) having a threshold of $1$.
\end{MyExa}

Iterated one-step symmetric-expansion of pair in $(\G_n \times \K_n) \backslash \Sy_n$ eventually yields a pair in $\Sy_n$. We formalize the idea:

\begin{MyDef} \label{Sexp}
Given a pair $(G_0,k_0)$ in $(\G_n \times \K_n) \backslash \Sy_n$, we construct a finite sequence $(G_0,k_0),(G_1,k_1),\cdots,(G_m,k_m)$ for some positive integer $m$, where $(G_l,k_l)$ is a one-step symmetric-expansion of $(G_{l-1},k_{l-1})$ and $(G_m,k_m)$ belongs to $\Sy_{n'}$ for some $n'$. We refer to $(G_m,k_m)$ as the symmetric-expansion of $(G,k)$.
\end{MyDef}

We note that the symmetric-expansion of $(G,k)$ is uniquely defined (up to isomorphism\footnote{ \label{isofoot} Two pairs $(G,k)$ and $(G',k')$ in $\G_n {\times} \K_n$ are said to be isomorphic if there exists a bijective map $\phi$ from $V(G)$ into $V(G')$ such that $ij \in E(G)$ if and only if $\phi_i \phi_j \in E(G')$ and $k_i = (k' \circ \phi)_i$ for $i$ in $V(G)$.}) regardless of the order the pivot nodes were chosen. However, we shall not write a formal proof for this fact, it would suffice to say that when we perform the one-step symmetric expansion of a pair $(G,k)$, we leave all neighborhoods and thresholds unchanged for all nodes different than the pivot.\footnote{ Given a pair $(G,k)$ in $(\G_n \times \K_n) \backslash \Sy_n$, it can be easily checked that if $(G',k')$ be the one-step symmetric-expansion of $(G,k)$, then $G$ is bipartite if and only if $G'$ is bipartite. Similarly, given a pair $(G,k)$ in $(\G_n \backslash \G^b_n) \times \K_n$, let $(G',k')$ be the bipartite-expansion of $(G,k)$, then $(G,k)$ belongs to $\Sy_{n}$ if and only if $(G',k')$ belongs to $\Sy_{2n}$.}\\

Our goal is to prove a bound (equal to 2) on the length of convergence cycles for all pairs $(G,k)$. We first establish that such a bound applies for a pair $(G,k)$ if it applies for either its symmetric-expansion (Lemmas \ref{OneStepM} and \ref{SymmetricM}) or its bipartite-expansion (Lemma \ref{BipartiteM}). We combine both expansions (Proposition \ref{CombineExp}), and thus show that, to get the general result, we need only prove the bound for symmetric models $(G,k)$ where $G$ is bipartite (Lemmas \ref{SBM} and \ref{SBM2}).

\begin{MyDef}
We define the set $\M$ to be a subset of $\bigcup_{n\geq1}\G_n{\times}\K_n$ such that $(G,k)$ in $\G_n{\times}\K_n$ belongs to $\M$ if and only if for every $C$ in $CYCLE_n(G,k)$, the cardinality of $C$ is less than or equal to $2$.
\end{MyDef}

What we ultimately show is that $\M = \bigcup_{n\geq1}\G_n{\times}\K_n$.

\begin{MyLem} \label{OneStepM}
Given a pair $(G,k)$ in $(\G_n \times \K_n) \backslash \Sy_n$, define $(G',k')$ to be the one-step symmetric-expansion of $(G,k)$. If $(G',k')$ belongs to $\M$ then $(G,k)$ belongs to $\M$.
\end{MyLem}
\begin{proof}
Let node $i$ be the \emph{pivot} node in the one-step symmetric-expansion of $(G,k)$ into $(G',k')$. We suppose $G'$ belongs to $\G_{n'}$, and consider the induced subgraph $H$ in $G'$ over the vertex set $\I_{n'}\backslash \I_n$. The subgraph $H$ necessarily consists (by construction of $G'$) of $d_i + 1$ connected components, each consisting of three vertices. We set integers $b_i$ and $w_i$ to be equal to $k_i$ and $d_i-k_i+1$ respectively, we then partition $\I_{n'}\backslash \I_n$ into $W_1,\cdots,W_{b_i},B_1,\cdots,B_{w_i}$ such that each partition contains the set of nodes in one of the connected components. Suppose that $(G',k')$ belongs to $\M$, we show that $(G,k)$ belongs to $\M$. Let $C$ be an element of $CYCLE_n(G,k)$ and suppose $C$ has cardinality greater or equal to 3. In particular, suppose $C$ is equal to $\{a_1,\cdots,a_m\}$ for $m\geq 3$ where $a_{l+1} = G_k a_l$ for $1 \leq l < m$ and $a_1 = G_k a_m$.

We define a map $\alpha$ from $\A_n$ into $\A_{n'}$ such that for $a$ in $\A_n$,
\begin{equation}
 \alpha a =
 	\left\{
		\begin{array}{ll}
			a & \text{on } \I_n \\
			\black & \text{on } B_1 \cup \cdots \cup B_{w_i}\\
			\white & \text{on } W_1 \cup \cdots \cup W_{b_i}\\
		\end{array}.\nonumber
	\right.
\end{equation}
First, the map $\alpha$ is clearly injective. Then, $\alpha a_1,\cdots, \alpha a_m$ are distinct elements of $\A_{n'}$. Second, it can be checked that
\begin{equation}
	\alpha (G_ka) = G'_{k'}(\alpha a). \nonumber
\end{equation}
To see this fact, notice that for every $b_1$ and $b_2$ in $\A_n$ every node in $\I_{n'}\backslash \I_n$ has the same color both in $\alpha b_1$ and $\alpha b_2$, and the same color both in $\alpha b_1$ and $G'_{k'}\alpha b_1$. Since every node in $\I_n$ other than the pivot $i$ keeps the same neighborhood and threshold, we need only show that $(\alpha G_ka)_i = (G'_{k'}\alpha a)_i$. To this end, node $i$ is $\black$ in $G_ka$ if and only if at least $k_i$ neighbors of $i$ in $G$ are $\black$ in $a$, and that is the case if and only if at least $k_i + w_i$ neighbors of $i$ in $G'$ are $\black$ in $\alpha a$, or equivalently at least $d^G(i) + 1 = (d^{G'}(i) + 1)/2$ neighbors of $i$ in $G'$ are $\black$ in $\alpha a$. Finally, it follows that $\alpha C = \{\alpha a_1,\cdots, \alpha a_m \}$ is a cycle in $CYCLE_{n'}(G',k')$ contradicting the fact that $m\geq 3$.
\end{proof}

\begin{MyLem} \label{SymmetricM}
Given a pair $(G,k)$ in $(\G_n \times \K_n) \backslash \Sy_n$, define $(G',k')$ to be the symmetric-expansion of $(G,k)$. If $(G',k')$ belongs to $\M$ then $(G,k)$ belongs to $\M$.
\end{MyLem}

\begin{proof}
Given that the pair $(G,k)$ belongs to $(\G_n \times \K_n) \backslash \Sy_n$, we can construct a finite sequence $(G_1,k_1), \cdots, (G_{m-1},k_{m-1})$ in such a way that $(G_1,k_1)$ is the one-step symmetric expansion of $(G,k)$, $(G',k')$ is the one-step symmetric-expansion of $(G_{m-1},k_{m-1})$ and $(G_{l},k_{l})$ is the one-step symmetric-expansion of $(G_{l-1},k_{l-1})$ for $1<l<m$. If $(G',k')$ belongs to $\M$, then $(G_{m-1},k_{m-1})$ belongs to $\M$ by Lemma \ref{OneStepM}. Recursively, it follows that $(G,k)$ belongs to $\M$.
\end{proof}

\begin{MyLem} \label{BipartiteM}
Given a pair $(G,k)$ in $(\G_n\backslash\G^b_n) \times \K_n$, define $(G',k')$ to be the bipartite-expansion of $(G,k)$. If $(G',k')$ belongs to $\M$ then $(G,k)$ belongs to $\M$.
\end{MyLem}
\begin{proof}
The graph $G'$ is bipartite and has vertex set $\I_{2n}$, let us partition $\I_{2n}$ into $\I_n$ and $\J$, then $(\I_n,\J)$ forms a 2-Partition with respect to $G'$ by construction. We define a bijection $\phi$ from $\J$ into $\I_n$ such that for $j_1$ and $j_2$ in $\J$, $j_1\phi(j_2) \in E'$ if and only if $\phi(j_1)\phi(j_2) \in E$. Given an action configuration in $\A_n$, we define the map $\alpha$ from $\A_n$ into $\A_{2n}$ such that for $a$ in $\A_n$ we have:
\begin{equation}
	\alpha a = (a , a \circ \phi). \nonumber
\end{equation}
It then follows that:
\begin{equation}
  G'_{k'} (\alpha a) = \alpha (G_k a). \nonumber
\end{equation}
The map $\alpha$ is clearly injective, it then follows that if $C = \{a_1, \cdots, a_m\}$ is a cycle in $CYCLE_n(G,k)$ with $m>2$, then $\alpha C = \{\alpha a_1, \cdots, \alpha a_m\}$ is a cycle in $CYCLE_n(G',k')$ contradicting the fact that $(G',k')$ belongs to $\M$.
\end{proof}

\begin{MyPro} \label{CombineExp}
Given a pair $(G,k)$ in $((\G_n\backslash\G^b_n) \times \K_n) \backslash \Sy_n$, we define $(G_1,k_1)$ and $(G_2,k_2)$ to be respectively the bipartite-expansion and the symmetric expansion of $(G,k)$. If $(G'_1,k'_1)$ is the symmetric-expansion of $(G_1,k_1)$ and $(G'_2,k'_2)$ is the bipartite-expansion of $(G_2,k_2)$, then $(G'_1,k'_1)$ is equal to $(G'_2,k'_2)$ (up to isomorphism\footnote{ See footnote \ref{isofoot}.}).
\end{MyPro}

\begin{proof}
It suffices to prove existence of a bijective map $\phi$ from $V(G'_2)$ to $V(G'_1)$, such that $ij \in E(G'_1)$ if and only if $\phi_i \phi_j \in E(G'_2)$ and $(k'_1)_i = (k'_2 \circ \phi)_i$ for $i$ in $V(G'_1)$. Such a map could be easily constructed following the expansion procedure. We omit the construction.
\end{proof}

\begin{MyLem} \label{SBM}
Given a pair $(G,k)$ in $((\G_n\backslash\G^b_n) \times \K_n) \backslash \Sy_n$, define $(G',k')$ to be the bipartite-expansion of $(G,k)$ and $(G'',k'')$ to be the symmetric-expansion of $(G',k')$. If $(G'',k'')$ belongs to $\M$ then $(G,k)$ belongs to $\M$.
\end{MyLem}

\begin{proof}
If $(G'',k'')$ belongs to $\M$, then $(G',k')$ belongs to $\M$ by Lemma \ref{SymmetricM}. It follows that $(G,k)$ belongs to $\M$ by Lemma \ref{BipartiteM}.
\end{proof}

\begin{MyLem} \label{SBM2}
The set $\M$ is equal to $\bigcup_{n\geq1}\G_n{\times}\K_n$ if and only if $\M$ contains $\Sy_n \cap (G^b_n {\times} \K_n)$ for every positive integer $n$.
\end{MyLem}

\begin{proof}
It is clear that if $\M=\bigcup_{n\geq1}\G_n{\times}\K_n$ then $\Sy_n \cap (G^b_n {\times} \K_n) \subset \M$. To prove the converse, given $(G,k)$ an element of $\bigcup_{n\geq1}\G_n{\times}\K_n$, we define $(G',k')$ to be the bipartite-expansion of $(G,k)$ and $(G'',k'')$ to be the symmetric-expansion of $(G',k')$. We have that $(G'',k'')$ belongs to $\Sy_n \cap (G^b_n {\times} \K_n)$. If $\M$ contains $\Sy_n \cap (G^b_n {\times} \K_n)$, it then follows by Lemma \ref{SBM} that $\M$ contains $(G,k)$.
\end{proof}

We proceed to show that $\M$ contains $\Sy_n \cap (G^b_n {\times} \K_n)$ for every positive integer $n$.

\begin{MyLem} \label{GeneralCond}
Given a pair $(G,k)$ in $G^b_n {\times} \K_n$, we consider a 2-Partition $(P_o,P_e)$ of $\I_n$ with respect to $G$. Then $(G,k)$ belongs to $\M$ if and only if for every $a$ in $\A_n$, there exists some integer $T$, such that $(G_k|_{P_e} G_k|_{P_o})^Ta = (G_k|_{P_e} G_k|_{P_o})^{T+1}a$.
\end{MyLem}

\begin{proof}
The statement of this lemma is equivalent to the statement of Proposition \ref{TreeSplit}.
\end{proof}

\begin{MyDef}[Conflict Link]
 Given $(G,a)$ in $\G_n{\times}\A_n$ with $G = (\I_n,E)$, we call a conflict link in $G$ with respect to $a$, an element $ij$ of $E$ such that $a_i$ and $a_j$ are not equal. We denote by $E_{c}^G(a)$ the set of all conflict links in $G$ with respect to $a$.
\end{MyDef}

\begin{MyLem} \label{ConvLemma}
The set $\M$ contains $\Sy_n \cap (G^b_n {\times} \K_n)$ for every positive integer $n$.
\end{MyLem}

\begin{proof}
Let $(G,k)$ in $\Sy_n \cap (G^b_n {\times} \K_n)$ be given and consider a 2-Partition $(P_o,P_e)$ of $\I_n$. Let $a$ be an action profile in $\A_n$. By Lemma \ref{GeneralCond} it would be enough to show that $(G_k|_{P_e} G_k|_{P_o})^Ta = (G_k|_{P_e} G_k|_{P_o})^{T+1}a$ for some non-negative integer $T$. In that case, it would be enough to prove that for every $b$ in $\A_n$,
\begin{equation}
  G_k|_{P_o} b \neq b \quad \text{ if and only if } \quad |E^G_c(G_k|_{P_o} b)| < |E^G_c(b)|, \nonumber
\end{equation}
and that similarly,
\begin{equation}
  G_k|_{P_e} b \neq b \quad \text{ if and only if } \quad |E^G_c(G_k|_{P_e} b)| < |E^G_c(b)|. \nonumber
\end{equation}
To show that, we state the following: for node $i$ in $\I_n$, $a_i \neq (G_k a)_i$ if and only if the majority of the players in $\N_i$ are not playing $a_i$, or equivalently, if and only if $i$ can decrease the number of conflict links by switching action.
\end{proof}

\begin{MyThm} \label{ConvThm}
For every positive integer $n$, every $(G,k)$ in $\G_n\times\K_n$ and every $C$ in $CYCLE_n(G,k)$, the cardinality of $C$ is less than or equal to $2$.
\end{MyThm}

\begin{proof}
The statement of the theorem is equivalent to $\M = \bigcup_{n\geq1}\G_n{\times}\K_n$. The fact that $\M = \bigcup_{n\geq1}\G_n{\times}\K_n$ follows directly from Lemma \ref{ConvLemma} and Lemma \ref{SBM2}.
\end{proof}

We will proceed to extend the result to the extension model. But before that, we perform a preparatory step to provide more insight on the problem. The primary model exhibits a (positive) monotonicity property in the following sense: a node is more likely to play $\black$ as the number of $\black$ playing neighbors increases. In what follows, we invert this property.

\subsection{Inverted Rules}

For a positive integer $n$, let $(G,k)$ be an element of $\G_n{\times}\K_n$. We define the map $\neg G_k$ from $\A_n$ into $\A_n$ such that
for $a\in A_n$, $(\neg G_ka)_i \neq (G_ka)_i$ for all $i$ in $\I_n$. Then $(\neg G_k a)_i$ is equal to $\black$ if and only if \emph{at most} $k_i - 1$ players are in $a^{-1}(\black)\cap \N_i$. Likewise, $(\neg G_ka)_i$ is equal to $\black$ if and only if \emph{at least} $d_i - k_i + 1$ players are in $a^{-1}(\white)\cap \N_i$. If we extend the definition of $CYCLE_n(G,k)$ to $\neg G_k$, and denote the set of cycles by $CYCLE_n(\neg G_k)$, we get the following theorem:

\begin{MyThm} \label{NegConvThm}
For every positive integer $n$, every $(G,k)$ in $\G_n\times\K_n$ and every $C$ in $CYCLE_n(\neg G_k)$, the cardinality of $C$ is less than or equal to $2$.
\end{MyThm}

\begin{proof}
For a positive integer $n$, let $(G,k)$ be an element of $\G_n{\times}\K_n$. We will simulate $\neg G_k$ using an instance of the primary model. For $i$ in $\I_n$, $(\neg G_k a)_i = \black$ if and only if at most $k_i -1$ players are in $a^{-1}(\black)\cap \N_i$. Since the bipartite-expansion procedure is only defined for non-bipartite graphs, for notational (and technical) convenience, we will assume the graph $G$ to be non-bipartite. The bipartite case may be easily recreated from this proof by considering a 2-Partition directly. Let $(G',k')$ be the bipartite expansion of $(G,k)$, then $G'$ is bipartite, has vertex set $\I_{2n}$ and edge set $E'$. Let us partition $\I_{2n}$ into $\I_n$ and $\J$. We define a bijection $\phi$ from $\J$ into $\I_n$ such that for $j_1$ and $j_2$ in $\J$, 
$j_1\phi(j_2) \in E'$ if and only if $\phi(j_1)\phi(j_2) \in E$. We define $\bar{k}$ such that $\bar{k}_i$ is equal to $k_i'$ for $i$ in $\J$ and equal to $d_i - k'_i +1$ for $i$ in $\I_n$. For $a$ in $\A_n$, we first define $\neg a$ to be the element of $\A_n$ such that $\neg a_i \neq a_i$ for all $i$ in $\I_n$, and then define $\alpha$ to be the map from $\A_n$ into $\A_{2n}$ such that for $a$ in $\A_n$:
\begin{equation}
 \alpha a  = ( a,(\neg a)\circ\phi).\nonumber
\end{equation}
This done, we get:
\begin{equation}
 G'_{\bar{k}} \alpha a = G'_{\bar{k}}( a,(\neg a)\circ\phi) = (\neg G_k a, (G_k a) \circ \phi). \nonumber
\end{equation}
To see this, for $i$ in $\I_n$, $(G'_{\bar{k}} \alpha a)_i$ is $\black$ if and only if at least $d_i - k_i +1$ players are in $\phi^{-1}((\neg a)^{-1}(\black))\cap\N_{G'}(i)$. Or equivalently, at least $d_i - k_i +1$ are in $\phi^{-1}(a^{-1}(\white)) \cap \N_{G'}(i)$. But at least $d_i - k_i +1$ are in $\phi^{-1}(a^{-1}(\white)) \cap \N_{G'}(i)$ if and only if at least $d_i - k_i +1$ are in $a^{-1}(\white) \cap \N_{G}(i)$ since $\N_{G}(i) = \phi(\N_{G'}(i))$. It then follows that $(G'_{\bar{k}} \alpha a)_i$ is $\black$ if and only if at most $k_i - 1$ players are in $a^{-1}(\black) \cap \N_{G}(i)$.

Let $C = \{a_1,\cdots,a_m\}$ be a cycle in $CYCLE_n(\neg G_k)$, then $\neg C = \{\neg a_1,\cdots,\neg a_m\}$ is a cycle in $CYCLE_n(G_k)$ and so 
$\{(a_1,(\neg a_1)\circ\phi), \cdots, (a_m,(\neg a_m)\circ\phi) \}$ is a cycle in $CYCLE_{2n}(G_{\bar{k}})$. Therefore, $m$ is less than or equal to $2$.
\end{proof}

This `inverted' rule may be viewed as having negative weights equal to $-1$ along all edges then making the thresholds negative. We proceed to study the extension model.

\subsection{Convergence Cycles for the Extension Model}

We extend the convergence cycle result to the extension model in three steps. We initially assume that no self loops are present. We extend the result first to the case where weights are allowed to be equal to either $1$ or $-1$ and then to arbitrary weights. We finally allow weighted self-loops to be present, and prove the theorem.\\

We begin by assuming that no self-loops are present. We then expand $\G_n$ to incorporate the weights. We define $\W_n$ to be the set of connected graph over the vertex set $\I_n$ with weighted edges. The set $\W_n$ is then isomorphic to the set of symmetric matrices in $\mathbb{R}^{n{\times}n}$ having zero elements along the diagonal. Given a weighted graph $W$ of $\W_n$, we denote the set of edges by $E(W)$, the weight on edge $ij$ in $E(W)$ is then denoted by $w_{ij}$. As a natural extension of previous definitions, given a pair $(W,q)$ we define the map $W_q$ from $\A_n$ into $\A_n$ such that $(W_q a)_i = \black$ if and only if $\sum_{j \in \N_i} w_{ij}\One_{ \{\black \} }(a_{j}) \geq q_i \sum_{j \in \N_i} w_{ij}$.\\

We define $\hat{\K}_n$ to be the set of maps from $\I_n$ into the set of integers $\mathbb{Z}$. Following a similar reason as in Section 3, without any loss in generality, we substitute the set $\Q_n$ with the set $\hat{\K}_n$ as defined; the thresholds can be negative in the extension model. For a pair $(W,k)$ in $\W_n \times \hat{\K}_n$, the map $W_k$ extends naturally from $W_q$ for the pair $(W,q)$ in $\W_n \times \Q_n$, in the way that for an action configuration $a$ in $\A_n$, $(W_k a)_i = \black$ if and only if $\sum_{j \in \N_i} w_{ij}\One_{ \{\black \} }(a_{j}) \geq k_i$. Similarly, we define $CYCLE_n(W,k)$ to extends naturally from Section 3.

We now restrict the weights to be in $\{-1,1\}$. Let $(W,k)$ be a pair in $\W_n{\times}\hat{\K}_n$, we define, for each $i$ in $\I_n$, $\N^+_i$ and $\N_i^-$ to be a partition of $\N_i$ such that $w_{ij}$ is positive only if $j$ belongs $\N^+_i$. It then follows that $(W_k a)_i$ is $\black$ if and only if $|a^{-1}(\black)\cap\N_i^+| - |a^{-1}(\black)\cap\N_i^-| \geq k_i$.

\begin{MyThm} \label{MixedThm}
 For any positive integer $n$, and any pair $(W,k)$ in $\W_n {\times} \hat{\K}_n$ where all weights are in $\{-1,1\}$, every cycle $C$ in $CYCLE_n(W,k)$ has cardinality less than or equal to 2.
\end{MyThm}

\begin{proof}
Given a pair $(W,k)$ in $\W_n {\times}\hat{\K}_n$ where all weights are in $\{-1,1\}$, we construct a graph $G'$ over $\I_{2n}$ with edge set $E'$ as follows. We partition $\I_{2n}$ into $\I_n$ and $\J$, and set up a bijection $\phi$ from $\J$ into $\I_n$. We assume that at least two weights are of opposite sign, otherwise nothing needs to be done. We let $E^+$ and $E^-$ be a partition of $E(W)$ such that the edges in $E^+$ and $E^-$ have positive and negative weights respectively. We then let $E'$ contain both $ij$ and $\phi(i)\phi(j)$ for all $ij$ in $E^+$, and $i\phi(j)$ for all $ij$ in $E^-$. Given a node $i$, we denote by $d_i^+$ and $d_i^-$ to be the number of edges having $i$ as an endpoint in $E^+$ and $E^-$ respectively. In this sense, $d_i = d_i^+ + d_i^-$. We construct $k'$ to be an element of $\hat{\K}_{2n}$ such that $k'_i$ equals $k_i + d^-_i$ for $i$ in $\I_n$ and equals $d^+_i - k_i + 1$ for $i$ in $\J$. Without any loss of generality, we assume that all nodes are valid. Therefore, $-d_i^- \leq k_i \leq d^+_i$ for all $i$ in $\I_n$.

For $a$ in $\A_n$, we first define $\neg a$ to be the element of $\A_n$ such that $\neg a_i \neq a_i$ for all $i$ in $\I_n$, and then define $\alpha$ to be the map from $\A_n$ into $\A_{2n}$ such that for $a$ in $\A_n$:
\begin{equation}
 \alpha a  = ( a,(\neg a)\circ\phi).\nonumber
\end{equation}
This done:
\begin{equation}
 G'_{k'} \alpha a = G'_{k'}( a,(\neg a)\circ\phi) = (W_k a, (\neg W_k a) \circ \phi).\nonumber
\end{equation}
To see this, for $i$ in $\I_n$, $(G'_{k'} \alpha a)_i = \black$ if and only if at least $k_i + d^-_i$ players are in $(\alpha a)^{-1}(\black)\cap\N^{G'}_i$, or equivalently: 
\begin{equation}
|a^{-1}(\black)\cap\N_i^+| + |a^{-1}(\white)\cap\N_i^-| \geq k_i + d_i^-.\nonumber
\end{equation}
But $|a^{-1}(\white)\cap\N_i^-| = d_i^- - |a^{-1}(\black)\cap\N_i^-|$. Therefore, $(G'_{k'} \alpha a)_i = \black$ if and only if $(W_k a)_i = \black$.

Similarly, for $i$ in $\J$, $(G'_{k'} \alpha a)_i = \black$ if and only if at least $d^+_i - k_i + 1$ players are in $(\alpha a)^{-1}(\black)\cap\N^{G'}_i$, or equivalently
\begin{equation}
|a^{-1}(\white)\cap\N_i^+| + |a^{-1}(\black)\cap\N_i^-| \geq d^+_i - k_i + 1.\nonumber
\end{equation}
The previous proposition is also equivalent to:
\begin{equation}
d^+_i - |\neg a^{-1}(\black)\cap\N_i^+| - |a^{-1}(\black)\cap\N_i^-| < k_i.\nonumber
\end{equation}
But, $d^+_i - |\neg a^{-1}(\black)\cap\N_i^+| = |a^{-1}(\black)\cap\N_i^+|$. It follows that $(G'_{k'} \alpha a)_i = (\neg W_k a)_i$.\\

Let $C = \{a_1,\cdots,a_m\}$ be a cycle in $CYCLE_n(W,k)$, then $\{(a_1,(\neg a_1)\circ\phi), \cdots, (a_m,(\neg a_m)\circ\phi) \}$ is a cycle in $CYCLE_{2n}(G'_{k'})$. Therefore, $m$ is less than or equal to $2$.
\end{proof}

We now extend the result to the extension model allowing non-equal weights on edges but no self-loops. We put no restrictions on the weights in $\W_n$ other than being non-zero reals. Without any loss of generality, we may assume the weights to be rational (see \cite{MyThesis}). In this case, we can multiply all weights by a common factor and have them be integers. This said, without any loss of generality, we assume the weights to be integers. 

\begin{MyThm} \label{ExtendedThm}
 For any positive integer $n$, and any pair $(W,k)$ in $\W_n {\times} \hat{\K}_n$, every cycle $C$ in $CYCLE_n(W,k)$ has cardinality less than or equal to 2.
\end{MyThm}

\begin{proof}  
  Given a pair $(W,k)$ in $\W_n {\times} \hat{\K}_n$ with integer weights, we construct an instance $(W',k')$ in $\W_{n'}{\times}\hat{\K}_{n'}$ with weights in $\{-1,+1\}$. We define $N$ to be equal to $\Pi_{ij \in E} w_{ij} $, and we consider the set of players $\I_{Nn}$. We now partition $\I_{Nn}$ into sets of $n$ players. Given the size of the object in hand, it would be appropriate to identify the partitions with the following $|E|$-dimensional space:
  \begin{equation}
  		\Omega = \Pi_{e \in E}[|w_e|], \nonumber
  \end{equation}
where $[|w_e|] = \{1,\cdots,|w_e|\}$.
Specifically, we define a map $P$ from $\Omega$	into $2^{\I_{Nn}}$ such that $|P(1,\cdots,1)| = \I_n$, $|P(\omega)| = n$ for all $\omega$, and $P(\omega) \cap P(\omega') = \emptyset$ for $\omega \neq \omega'$. Then the collection $\{P(\omega) : \omega \in \Omega\}$ is a partition of $\I_{Nn}$. For each $\omega$, we define a bijection $\phi_{\omega}$ from $\I_n$ into $P(\omega)$, such that $\phi_{(1,\cdots,1)}$ is the identity map. We now define a graph $G$ over the vertex set $\I_{Nn}$ with edge set $E = E^+ \cup E^-$. Then for each $ij$ that belongs to $E(W)$, for all $\omega$ in $\Pi_{e \in E\backslash{ij}}[|w_e|]$, $m$ and $m'$ in $[|w_{ij}|]$ (non-necessarily distinct), we let $\phi_{(\omega,m)}(i)\phi_{(\omega,m')}(j)$ belong to $E^+$ if $w_{ij}$ is positive and belong to $E^-$ if $w_{ij}$ is negative. We give all edges in $E^+$ and $E^-$ weights of $+1$ and $-1$ respectively. We finally construct a threshold distribution $k'$ on $\I_n$ in such a way that $k'$ equal $k \circ \phi_{\omega}^{-1}$ on $P(\omega)$ for all $\omega$ in $\Omega$.

To put a note on the construction, given any node $i$ in $\I_n$, suppose $j$ is a neighbor of $i$, with weight $w_{ij}$ on the edge. For any $\omega$ in $\Omega$, $\phi_{\omega}(i)$ has threshold $k_i$ and is connected to exactly $|w_{ij}|$ nodes having thresholds $k_j$. This large graph is interconnected in such a way that if we restrict the space of action configurations accordingly, the update are locally equivalent.

To this end, let us define the extension map $\alpha$ from $\A_n$ into $\A_{Nn}$ in such a way that for all $a$ in $\A_n$, for all $\omega$ and $i$ in $P(\omega)$,
\begin{equation}
 (\alpha a)_i = a_{\phi_\omega^{-1}(i)}.  \nonumber
\end{equation}
 It can be checked that for all $a$ in $\A_n$:
\begin{equation}
  \alpha (W_k a)  = W'_{k'} (\alpha a) \nonumber
\end{equation}
The map $\alpha$ is injective, and following the same reasoning as in the proof of Lemma \ref{BipartiteM}, the result then follows since any cycle in $CYCLE_{Nn}(W',k')$ has cardinality at most equal to 2 by Theorem \ref{MixedThm}. 
\end{proof}

We finally allow weighted self-loops in the network. We define $\hat{\W}_n$ to be the set of connected graph with self-loops allowed over the vertex set $\I_n$ with weighted edges. The set $\hat{\W}_n$ is then isomorphic to the set of symmetric matrices in $\mathbb{R}^{n{\times}n}$. Again, given a weighted graph $\hat{W}$ of $\hat{\W}_n$, we denote the set of edges by $E(\hat{W})$, the weight on edge $ij$ in $E(\hat{W})$ is then denoted by $w_{ij}$. As a natural extension of previous definitions, given a pair $(\hat{W},q)$ in $\hat{\W}_n{\times}\Q_n$ we define the map $\hat{W}_q$ from $\A_n$ into $\A_n$ such that $(\hat{W}_q a)_i = \black$ if and only if $\sum_{j \in \N_i} w_{ij}\One_{ \{\black \} }(a_{j}) \geq q_i \sum_{j \in \N_i} w_{ij}$ where $\N_i$ may possibly contain player $i$.\\

Without any loss of generality, we may assume the weights to be integers. Moreover, we substitute the set $\Q_n$ with the set $\hat{\K}_n$ as defined and therefore, given a pair $(\hat{W},k)$ in $\hat{\W}_n{\times}\hat{\K}_n$, we get that $(\hat{W}_k a)_i = \black$ if and only if $\sum_{j \in \N_i} w_{ij}\One_{ \{\black \} }(a_{j}) \geq k_i$ where $\N_i$ may possibly contain player $i$.

\begin{MyThm} \label{ExtendedThmWithSelfLoops}
 For any positive integer $n$, and any pair $(\hat{W},k)$ in $\hat{\W}_n {\times} \hat{\K}_n$, every cycle $C$ in $CYCLE_n(\hat{W},k)$ has cardinality less than or equal to 2.
\end{MyThm}

\begin{proof}
 Given a pair $(\hat{W},k)$ in $\hat{\W}_n {\times} \K_n$, we construct a pair as follows $(W',k')$ in $\W_{2n} {\times} \K_{2n}$ as follows. Partition $\I_{2n}$ into $\I_n$ and $\J$, set up a bijection $\phi$ from $\J$ into $\I_n$ and let $E(W')$ contain $ij$ and $\phi(i)\phi(j)$ for all distinct players $i$ and $j$, and $i\phi(i)$ whenever node $i$ has a self-loop. We set $w'_{ij} = w'_{\phi(i)\phi(j)} = w_{ij}$ for all $i$ and $j$ distinct and $w'_{i\phi(i)} = w_{ii}$ for all $i$. Finally, $k'$ equals $k$ on $\I_n$ and equals $k\circ \phi$ on $\J$. In this case, each cycle in $CYCLE_n(\hat{W},k)$ is of cardinality at most $2$ if each cycle in $CYCLE_{2n}(W',k')$ is of cardinality at most $2$. The result then follows from Theorem \ref{ExtendedThm}.
\end{proof}

We proceed to discuss convergence time in the following section.

\section{On Convergence Time}

We proceed to study the following problem: given a graph $G$ in $\G_n$, a threshold distribution $k$ in $\K_n$ and an initial action configuration $a$, how many times do we need to iteratively apply $G_k$ on $a$ to reach some cycle $C$ in $CYCLE_n(G,k)$? Recall that for every positive integer $n$, and every $(G,k,a)$ in $\G_n{\times}\K_n{\times}\A_n$, we define $\delta_n(G,k,a)$ to be equal to the smallest non-negative integer $T$ such that there exists a cycle $C$ in $CYCLE_n(G,k)$ and $b$ in $C$ with $G_k^Ta = b$. The quantity $\delta_n(G,k,a)$ denotes to the minimal number of iterations needed until a given action configuration $a$ reaches a cycle, when iteratively applying $G_k$. We refer to $\delta_n(G,k,a)$ as the \emph{convergence time} from $a$ under $G_k$. We begin by showing that there exists some positive integer $c$, such that for every positive integer $n$, and every $(G,k,a)$ in $\G_n{\times}\K_n{\times}\A_n$, the convergence time $\delta_n(G,k,a)$ is less than or equal to $cn^2$. We then proceed to improve the bound to be linear in the size of the network for some graph structure cases. Formally, for all positive integers $n$, and every $(G,k,a)$ in $\G_n{\times}\K_n{\times}\A_n$ where $G$ is a cycle graph (with an even number of nodes) or a tree, the convergence time $\delta_n(G,k,a)$ is less than or equal to $n$. Such a result also holds if the network structure is a complete graph, we refer the reader to \cite{MyThesis}. We will only consider the primary model in this section.

\subsection{Quadratic Time over General Graphs}

The proofs of Lemmas \ref{BipartiteM}, \ref{OneStepM} and \ref{SymmetricM} illustrate how we can simulate a pair $(G,k)$ on its symmetric-expansion $(G',k')$ and its bipartite-expansion $(G'',k'')$ by restricting the set of initial action profiles on the expansions.  It also follows (from the proofs) that any upper-bound on the convergence time of $(G',k')$ and $(G'',k'')$ is also an upper-bound on the convergence time of $(G,k)$. For the case of symmetric models where the graph is bipartite, the proof of Lemma \ref{ConvLemma} shows that the convergence time is upper-bounded by the number of initial conflict links in the graph i.e.\ (at worst) by a function quadratic in the size of the graph. As for any pair $(G,k)$, a symmetric-expansion followed by a bipartite-expansion cannot add `too many' edges. This yields a convergence time for any $(G,k)$ bounded by a function quadratic in the size of $G$.

\begin{MyThm} For some positive integer $c$, for all positive integers $n$, and every $(G,k,a)$ in $\G_n{\times}\K_n{\times}\A_n$, the convergence time $\delta_n(G,k,a)$ is less than or equal to $cn^2$.
\end{MyThm}

\begin{proof}
Given a positive integer $n$, let $(G,k)$ be a point in $\G_n\times\K_n$, let $(G',k')$ be the symmetric-expansion of $(G,k)$ in $\Sy_{n'}$ , and let $(G'',k'')$ be the bipartite-expansion of $(G',k')$ in $\Sy_{2n'}$. We have not formally defined bipartite-expansion for bipartite graphs and symmetric-expansion for elements of $\Sy_n$. Although an extension is trivial, we will assume that $G$ is non-bipartite and $(G,k)$ does not belong to $\Sy_n$. The other cases then follow directly from the proof by skipping the expansion steps.
We have the following fact:
 \begin{equation}
 		\delta_n{{(G,k)}} \leq \delta_{n'}{{(G',k')}} \leq \delta_{2n'}{{(G'',k'')}}. \nonumber
 \end{equation}
Moreover, we have that:
\begin{equation} 
	\delta_{2n'}{{(G'',k'')}} \leq \max_{a \in \A_{2n'}} {|E_{c}^{G''}(a)|}.\nonumber
\end{equation}
The previous statement follows from the fact that for any 2-Partition $(P_o,P_e)$ of $\I_{2n'}$ with respect to $G''$, $G''_k|_{P_o} b \neq b$  if and only if  $E_c^{G''}(G''_k|_{P_o} b) < E_c(b)$ and $G''_k|_{P_e} b \neq b$ if and only if  $E_c^{G''}(G''_k|_{P_e} b) < E_c(b)$.
Additionally for all $a$ in $\A_{2n'}$ we have:
\begin{equation}
	|E_{c}^{G''}(a)| \leq |E''| \leq {2|E'|} \leq 2[|E| + 3\sum_{i\in \I_{n}}d_i + 1], \nonumber
\end{equation}
where $E'$ and $E''$ denotes the set of edges of $G'$ and $G''$ respectively. 
Finally:
\begin{equation}
  \sum_{i\in \I_{n}}d_i+1 = 2|E| + n. \nonumber
\end{equation}
The result follows.
\end{proof}

The constant $c$ in the theorem statement can be optimized, but it is of no interest. Instead it would be interesting to prove a bound below quadratic. We may notice from the proof above is that if the graph has bounded degrees, the convergence time is less than a linear function of the size of the network.\\

We turn back to the cases of cycle graphs and trees to derive tighter upper bounds on the convergence time.

\subsection{Linear Time over Cycle Graphs}

We restrict the analysis in the subsection to \emph{even} positive integers $n$. In this case, every 2-regular connected graph in $\G_n$ is bipartite and we make use of the bipartite property. Let $G$ be cycle graph in $\G_n$. Recall from subsection \ref{CRing} that we defined $s$ and $p$ to be maps from $\I_n$ into $\I_n$ (we refer to them successor and predecessor) such that for node $i$ in $\I_n$, $i$ and $s_i$ are neighbors, $i$ and $p_i$ are neighbors and $(sp)_i = (ps)_i = i$. In this setting, $(ss)_i$ refers to the successor of the successor of node $i$ and is denoted as $s^2_i$. Recursively, the notation $s^m_i$, where $m$ is some non-negative integer, denotes the node obtained by iteratively applying ($m$ times) the successor function $s$ on $i$. A similar notation holds for the predecessor function $p$. We now pick a player $i$ in $\I_n$, and consider the subset $\{s^{2m}_i : m \geq 0\}$ of $\I_n$. First, $\{s^{2m}_i : m \geq 0\}$ is not equal to $\I_n$, this follows from the fact that $n$ is even. Furthermore, $\{s^{2m}_i : m \geq 0\}$ is equal to $\{p^{2m}_i : m \geq 0\}$. This implies that the update rule over two time steps of player $i$ depends only on the information available in the actions taken by the players in $\{s^{2m}_i : m \geq 0\}$. We consider a 2-Partition $(P_o,P_e)$ of $\I_n$ with respect to $G$, and point out that the sequence $a, G^2_ka, G^4_ka, \cdots$ is constant after time step $2T$ if and only if both sequences $a{\upharpoonright}P_o, G^2_ka{\upharpoonright}P_o, G^4_ka{\upharpoonright}P_o, \cdots$ and $a{\upharpoonright}P_e, G^2_ka{\upharpoonright}P_e, G^4_ka{\upharpoonright}P_e, \cdots$ are constant after time step $2T$.

\begin{MyPro}
  For all positive even integers $n$, and every $(G,k,a)$ in $\G_n{\times}\K_n{\times}\A_n$ where $G$ is 2-regular, the convergence time $\delta_n(G,k,a)$ is less than or equal to $n$.
\end{MyPro}

\begin{proof}
For any 2-Partition $(P_o,P_e)$ of $\I_n$ with respect to $G$, let us then consider the sequence $a{\upharpoonright}P_o, G^2_ka{\upharpoonright}P_o, G^4_ka{\upharpoonright}P_o, \cdots$. We define $T$ to be the minimum integer $m$ such that $G^{2m}_k a{\upharpoonright}P_o = G^{2m + 2}_k a{\upharpoonright}P_o$. We claim that the integer $T$ is less than or equal to the cardinality of $P_o$. Let $S(b)$ be the number of players in $P_o$ playing the strong assignment in action configurations $b$. If $b \neq G^2_kb$ then $S(b) < S(G^2_kb)$ for all action configuration $b$. But $G^{2m}_k a{\upharpoonright}P_o \neq G^{2m + 2}_k a{\upharpoonright}P_o$ for all $m<T$ and so $P_o \geq S(G^{2T}_ka) \geq S(a) + T$. But $P_o = n/2$, and the result follows. 
\end{proof}

In other words, as long as the sequence $a{\upharpoonright}P_o, G^2_ka{\upharpoonright}P_o, G^4_ka{\upharpoonright}P_o, \cdots$ is not constant, one player is switching his action over two time steps. However, every player has a strong assignment (non-valid nodes trivially have a strong assignment), and so each player is allowed to flip only once over two time steps if ever. But $P_o$ contains only $n/2$ players. So after $2(n/2)$ time step, no player is able to flip over two time steps.

\subsection{Linear Time over Trees}

We now show that a similar result holds for trees. In this subsection, the letter $T$ shall always be used to denote trees, and never time as was done sometimes in previous sections. We prove a lemma, and then derive our result from it. We begin by a definition.

\begin{MyDef}
Given a triplet $(G,k,a)$ in $\G^b_n{\times}\K_n{\times}\A_n$ and a 2-Partition $(P_o,P_e)$ of $\I_n$ with respect to $G$, we say that $a{\upharpoonright}P_o$ is reachable in $(G,k)$ if there exists $a'$ in $\A_n$ such that $a{\upharpoonright}P_o = (G_k|_{P_o}a'){\upharpoonright}{P_o}$. In this case, we say that $a'{\upharpoonright}P_e$ induces $a{\upharpoonright}P_o$. Similarly, $a{\upharpoonright}P_e$ is reachable in $(G,k)$ if there exists $a'$ in $\A_n$ such that $a{\upharpoonright}P_e = (G_k|_{P_e}a'){\upharpoonright}{P_e}$. And again, we say that $a'{\upharpoonright}P_o$ induces $a{\upharpoonright}P_e$.
\end{MyDef}

\begin{MyPro} \label{pro1}
Given a triplet $(G,k,a)$ in $\G^b_n{\times}\K_n{\times}\A_n$ and a 2-Partition $(P_o,P_e)$ of $\I_n$ with respect to $T$. If node $i$ in $P_e$ is non-valid and $a{\upharpoonright}P_e$ is reachable in $(G,k)$, then $((G_k|_{P_e} G_k|_{P_o})^m a)_i = a_i$ for all non-negative integers $m$. 
\end{MyPro}

\begin{proof}
The proposition is rather trivial and follows from the definition of non-validity and reachability.
\end{proof}

\begin{MyPro} \label{pro2}
Given a triplet $(T,k)$ in $\G^b_n{\times}\K_n$ where $T$ is a tree, and a 2-Partition $(P_o,P_e)$ of $\I_n$ with respect to $T$. Pick a node $r$ to be the root of $T$, then if player $i$ in $P_e$ has only leaves as children (with respect to $r$) and $1<k_i<d_i$ then $((T_k|_{P_e} T_k|_{P_o})^m a)_i = a_i$ for all non-negative integers $m$.
\end{MyPro}

\begin{proof}
In the case where $1<k_i<d_i$, both actions ($\white$ and $\black$) are strong assignments for node~$i$.
\end{proof}

\begin{MyPro} \label{pro3}
Given a triplet $(T,k,a)$ in $\G^b_n{\times}\K_n{\times}\A_n$ where $T$ is a tree, and a 2-partition $(P_o,P_e)$ of $\I_n$ with respect to $T$. Suppose $a{\upharpoonright}P_o$ is reachable and $a{\upharpoonright}P_o$ induces $a{\upharpoonright}P_e$ both in $(T,k)$. Then if $P_e$ consists of only one node $i$, then $((T_k|_{P_e}T_k|_{P_o})^{m}a)_i = a_i$ for all non-negative integers~$m$.
\end{MyPro}

\begin{proof}
We suppose that the node in $P_e$ is valid, otherwise nothing is to be done. Furthermore, without any loss of generality we may assume the nodes in $P_o$ to also be valid, otherwise they will never switch actions given that $a{\upharpoonright}P_o$ is reachable and we may remove them from the network. It then follows that both $\black$ and $\white$ are strong assignments for the node in $P_e$.
\end{proof}

\begin{MyLem} \label{Dlemma}
For every positive integer $n$, given a triplet $(T,k,a)$ in $\G^b_n{\times}\K_n{\times}\A_n$ where $T$ is a tree and a 2-Partition $(P_o,P_e)$ of $\I_n$ with respect to $T$. If $a{\upharpoonright}P_o$ is reachable and $a{\upharpoonright}P_o$ induces $a{\upharpoonright}P_e$ (both) in $(T,k)$, then there exists a player $i$ in $P_e$, such that $((T_k|_{P_e}T_k|_{P_o})^{m}a)_i = a_i$ for all non-negative integers~$m$.
\end{MyLem}

\begin{proof}
Given the nice structure the tree possesses, there are several ways we can perform the induction. We proceed by induction on the number of nodes in the tree. We start with the base case that refers to a two node graph with a single edge. In this setting, there are only six cases of possible threshold distribution. It is fairly straightforward to exhaustively check them, so we omit the proof for $n=2$. We suppose that the statement holds for graphs with $n$ nodes, and we show that it holds for graphs with $n+1$ nodes.

We pick a triplet $(T,k,a)$ in $\G^b_n{\times}\K_n{\times}\A_n$ where $T(\I_n,E)$ is a tree, and a 2-Partition $(P_o,P_e)$ of $\I_n$ with respect to $T$. We suppose that $a{\upharpoonright}P_o$ is reachable and $a{\upharpoonright}P_o$ induces $a{\upharpoonright}P_e$ both in $(T,k)$. If there exists a player in $P_e$ that is non-valid with respect to $(G,k)$, the statement trivially holds by Proposition \ref{pro1}. We will assume that all nodes $P_e$ are valid nodes. We may also assume that all nodes in $P_o$ are valid nodes, otherwise they would never change actions and so can be removed. We pick a node $r$ in $P_o$ to be the root of the tree. If there exists some player $i$ in $P_e$ that has only leaves as children (with respect to $r$) and $1<k_i<d_i$, then the statement holds by Proposition \ref{pro2}. Then we will assume that no such player exists. Moreover, if a player in $P_e$ is playing a strong assignment, then the statement trivially holds. So, we assume that no player is playing a strong assignment. And finally, if $P_e$ contains only one node, then the statement holds by Proposition \ref{pro3}. We then assume that $P_e$ contains at least two players.

We argue by contradiction. Suppose that for every player $i$ in $P_e$, there exits a positive integer $M_i$ such that  $((T_k|_{P_e} T_k|_{P_o})^{M_i}a)_i \neq a_i$ and $((T_k|_{P_e}T_k|_{P_o})^{m}a)_i = a_i$ for $m<M_i$. We set $M_f$ to be $\max_i M_i$, and pick a node $e$ in $P_e$ such that $M_e = M_f$. We consider an edge $eo$ in $E$ such that $a_o = a_e$ and if $e$ has only leaves as children, the connected components of $(\I_n,E\backslash\{oe\})$ contain at least one node in $P_e$. Such an edge always exists given what we assumed earlier.

We then construct a pair $(T',k')$ as follows. Define $T'$ to be the connected component of the graph with vertex set $\I_n$ and edge set $E \backslash \{oe\}$ not containing $e$. Set $k'$ to be equal to $k$ on $V(T') \backslash\{o\}$, $k$ on $\{o\}$ if $a_e = \white$, and $(k-1)\vee0$ on $\{o\}$ if $a_e = \black$.

One can check that $a{\upharpoonright}(P_o\cap V(T'))$ is reachable and $a{\upharpoonright}(P_o\cap V(T'))$ induces $a{\upharpoonright}(P_e \cap V(T'))$ both in $(T',k')$. In this case, there exists a player $i$ in $P_e \cap V(T')$, such that $((T'_{k'}|_{P_e}T'_{k'}|_{P_o})^{m}a)_i = a_i$ for all non-negative integers $m$. Then $((T_{k}|_{P_e} T_{k}|_{P_o})^{m}a)_i = a_i$ for all positive integer~$m$ such that
\begin{equation}
  ((T_{k}|_{P_e} T_{k}|_{P_o})^{m-1}a)_e = a_e, \nonumber
\end{equation}
that is for $m \leq M_f$. Then, player $i$ \emph{can} only flip in $(T,k)$ after $M_f$, contradicting the definition of $M_f$.
\end{proof}

\begin{MyThm}
For all positive integers $n$, and every $(T,k,a)$ in $\G_n{\times}\K_n{\times}\A_n$ where $T$ is a tree, the convergence time $\delta_n(T,k,a)$ is less than or equal to $n$.
\end{MyThm}

\begin{proof}
Let $T$ be a tree in $\G_n$, and consider a 2-Partition $(P_o,P_e)$ of $\I_n$ with respect to $T$ such that $|P_e| \leq |P_o|$. For any $k$ in $\K_n$ and $a$ in $\A_n$, if we consider $(a^o,a^e) = (T_ka{\upharpoonright}P_o,T^2_ka{\upharpoonright}P_e)$, then $a^o$ is reachable, and $a^o$ induces $a^e$. Then by Lemma \ref{Dlemma}, there exists at least one node $i$ in $P_e$ such that $((T_k|_{P_e}T_k|_{P_o})^{m}(a^o,a^e))_i = a^e_i$ for all non-negative integers~$m$. Let $T'$ be a connected component of the induced subgraph of $T$ over $\I_n \backslash \{i\}$ such that $|V(T') \cap P_e|$ is maximal. Define $k'$ to be equal to $k{\upharpoonright}V(T')$ on $V(T')\backslash \N_i$, equal to $k$ on $V(T')\cap \N_i$ if $a^e_i = \white$, and equal to $(k-1)\vee 0$ on $V(T')\cap \N_i$ if $a^e_i = \black$. Then,
\begin{equation}
 \delta_n(T,k,a) = 2 + \delta(T,k,T^2_ka) = 2 + \delta(T',k',(T^2_ka){\upharpoonright}V(T')) \leq 2|P_e| \nonumber
\end{equation}
by successive application of Lemma \ref{Dlemma}. But since $|P_e| \leq |P_o|$, we get $|P_e| \leq n/2$ and the result follows.
\end{proof}

Notice that leaving the bound at $2|P_e|$ in the preceding proof gives a better bound that is equal to twice the size of a minimal vertex cover. Unfortunately, Lemma \ref{Dlemma} does not hold over all bipartite graphs (see \cite{MyThesis}), however this does not disprove fast convergence time. We end this section with a conjecture: the convergence time $\delta_n(G,k,a)$ is less than or equal to $n$ whenever $G$ is bipartite. In this case, $\delta_n(G,k,a)$ will necessarily be less than or equal to $2n$ when $G$ is non-bipartite, and this bound would be tight. We now proceed to characterize the number of limiting cycles.


\section{On the Complexity of Counting}

So far, we have been dealing with bounds that are uniform over all graphs, all thresholds and all action configurations. The natural coming step would be to find bounds on the number of cycles (fixed-points and non-degenerate cycles), the number of fixed points and the number of non-degenerate cycles for all graphs $G$ in $\G_n$ and all threshold distributions $k$ in $\K_n$. We let $\overline{F}$ and $\underline{F}$ be respectively the maximum and minimum number of fixed-points over all pair $(G,k)$ in $\G_n{\times}\K_n$. Likewise, let $\overline{C}$ and $\underline{C}$ be respectively the maximum and minimum number of non-degenerate cycles over all pair $(G,k)$ in $\G_n{\times}\K_n$.

\begin{MyPro}
The lower bounds $\underline{F}$ and $\underline{C}$ are upper bounded by $2$ and $0$ respectively.
\end{MyPro}

\begin{proof}
To prove the proposition, it would be enough to consider any 2-connected regular graph of $n$ players where $n$ is odd, and provide each player with a threshold equal to $1$. All players playing $\black$ and all players playing $\white$ are the only fixed-points. No non-degenerate cycles exist at the limit.
\end{proof}

\begin{MyPro}
 The upper bounds $\overline{F}$ and $\overline{C}$ are lower bounded by $2^{n/3}-1$.
\end{MyPro}

\begin{proof}
 Consider the 2-connected regular graph of $n$ players where $n$ is a multiple of $3$ not equal to $3$. Assign $n/3$ nodes a threshold of $2$ and $2n/3$ nodes a threshold of $1$ in such a way that each node of threshold $1$ has exactly one neighbor of threshold $1$ connected to it. In this case, we have $n/3$ pairs of neighbors having thresholds of $1$, and so we can construct at least $2^{n/3}$ fixed points where the neighbors in each pair are either playing both $\black$ or both $\white$. We can similarly construct at least $2^{n/3} - 1$ non-degenerate cycles by having for each pair of such neighbors, either both neighbors as $\white$ or exactly one of the neighbors as $\black$.
\end{proof}

We are not concerned about exact bounds, those claim serve only to show that we are dealing with a rather wide range of number of limiting outcomes. This said, we will study whether we can have an arbitrarily good characterization of the count. Instead of providing bounds, we will study how tractable is it to count equivalence classes, fixed points and non-degenerate cycles. Ultimately, we show that those counting problems are \#P-Complete. We refer the reader the Appendix \ref{ComplexityAppendix} for a review of the needed concepts for this section.\\

\subsection{The Complexity of Counting Cycles and Fixed Points}

We characterize the number of equivalence classes, fixed points and cycles of length two. To this end, we define three counting problems: one for each.

\begin{MyDef}
The counting problem $\#CYCLE$ takes $<n,G,k>$ as input, where $n$ is a positive integer and $(G,k)$ belongs to $\G_n\times\K_n$, and outputs the cardinality of $CYCLE_n(G,k)$
\end{MyDef}

\begin{MyDef}
The counting problem $\#FIX$ takes $<n,G,k>$ as input, where $n$ is a positive integer and $(G,k)$ belongs to $\G_n\times\K_n$, and outputs the cardinality of $\{C \in CYCLE_n(G,k) : |C| = 1\}$
\end{MyDef}

\begin{MyDef}
The counting problem $\#2CYCLE$ takes $<n,G,k>$ as input, where $n$ is a positive integer and $(G,k)$ belongs to $\G_n\times\K_n$, and outputs the cardinality of $\{C \in CYCLE_n(G,k) : |C| = 2\}$
\end{MyDef}

Note that in this setting, given an input $<n,G,k>$, the output of $\#CYCLE$ is equal to the sum of the outputs of both $\#FIX$ and $\#2CYCLE$ when fed with the same input. Referring to the networked coordination game defined in the primary model, $\#FIX$ refers to counting the number of pure Nash equilibria of the networked game. We show the following:

\begin{MyThm} \label{CCycle}
 $\#CYCLE$ is \#P-Complete.
\end{MyThm}

\begin{MyThm} \label{CFix}
 $\#FIX$ is \#P-Complete.
\end{MyThm}

\begin{MyThm} \label{CTwoCycle}
 $\#2CYCLE$ is \#P-Complete.
\end{MyThm}

One has to be subtle towards what such result entails. This result does not imply that no characterization of the number of cycles is possible whatsoever, but rather that we would be unable to get an arbitrarily refined characterization of that number.\\

For technical insight, we may further note that no result in those three implies another, and no two results imply the third (or at least that no deduction may be made simply from the statements above with no additional information whatsoever). Specifically, if it is hard to count the number of fixed points, counting the number of cycles is not necessarily hard because of set inclusion. As a quick example, consider counting the number of total action configurations, surely this set includes the number of fixed points. However, counting them is trivial given the network size. Similarly, no hardness can directly be deduced by the fact that $\#CYCLE$ outputs the sum of $\#FIX$ and $\#2CYCLE$ when all the counting problems are fed with the same input. To illustrate quickly, consider counting the number of non-fixed point action configurations, this problem is hard since counting the number of fixed points is hard itself, however counting the number of fixed-points and non-fixed-points is again trivial given the size of the network.

We will prove our results as follows, we will restrict our input to only bipartite graphs. Within this restricted space, those three problems share a common ground that will be stated in two lemmas to follow. Mainly, either all of them are hard, or none of them is hard. We then prove the three theorems in one instance by showing that with restricted inputs one of the problems is \#P-Hard.

Building on the framework defined in section 4, we begin by this crucial observation.

\begin{MyLem} \label{FPEquivalence}
Let $n$ be a positive integer, $(G,k)$ be in $\G^b_n{\times}\K_n$ with $G = (\I_n,E)$ and $(P_o,P_e)$ a 2-Partition of $\I_n$ with respect to $G$.
For $a$ in $\A_n$, $a$ is a fixed point of $G_k$ if and only if $a$ is a fixed point of $G_k|_{P_e}G_k|_{P_o}$.
\end{MyLem}

\begin{proof}
 It is clear that, if $G_ka = a$, then $G_k|_{P_o}a = G_k|_{P_e}a = a$ and so $a$ is fixed point of $G_k|_{P_e} G_k|_{P_o}$.
 To show the converse, suppose that $G_k|_{P_e} G_k|_{P_o}a = a$, then $(G_k|_{P_e}G_k|_{P_o}a)_i = a_i$ for all $i$ in $\I_n$. If $i$ is in $P_o$, it follows that $(G_k|_{P_o}a)_i = a_i$ since $G_k|_{P_e}$ cannot modify $a_i$, and so $G_k|_{P_o}(a) = a$ since $(G_k|_{P_o}a)_i = a_i$ for $i$ in $P_e$. It follows that $G_k|_{P_e}(a) = a$. But 
 $(G_ka)_i = (G_k|_{P_o}a)_i$ if $i$ is in $P_o$, and $(G_ka)_i = (G_k|_{P_e}a)_i$ if $i$ is in $P_e$. Therefore, $G_k(a)=a$.
\end{proof}

With this in mind, we may proceed to the following curcial lemma.

\begin{MyLem} \label{FunctionOFix}
Let $(G,k)$ be a pair in $\G^b_n{\times}\K_n$, and let $F$ be equal to $|FIX_n(G,k)|$, we have $|CYCLE_n(G,k)| = F(F-1)/2 + F$.
\end{MyLem}

\begin{proof}
 Let $(G,k)$ be in $\G^b_n{\times}\K_n$ with $G = (\I_n,E)$ and $(P_o,P_e)$ be a 2-Partition of $\I_n$ with respect to $G$. To prove the result, it would be enough to construct a bijection from the set of non-degenerate cycles in $CYCLE_n(G,k)$ (having a cardinality equal to 2) to $\{(a',a'') \in FIX_n(G,k)^2 : a' \neq a''\}$. To this end, consider two distinct elements $a'$ and $a''$ in $FIX_n(G,k)$. By Lemma \ref{FPEquivalence}, we know that $a'$ and $a''$ are fixed points of $G_k|_{P_e}G_k|_{P_o}$. Construct $a_1 = (a''{\upharpoonright}P_o,a'{\upharpoonright}P_e)$ and $a_2 = (a'{\upharpoonright}P_o,a''{\upharpoonright}P_e)$. We get that $\{a_1, a_2 \}$ is a non-degenerate cycle in $CYCLE_n(G,k)$. Indeed, $G_k a_1 = a_2$ and $G_k a_2 = a_1$. Finally, clearly the map \[(a',a'') \mapsto ((a''{\upharpoonright}P_o,a'{\upharpoonright}P_e),(a'{\upharpoonright}P_o,a''{\upharpoonright}P_e))\] is bijective. The result follows from the fact a non-degenerate cycle in $CYCLE_n(G,k)$ is an unordered pair of action configurations, and that $CYCLE_n(G,k)$ contains both non-degenerate cycles and fixed-points.
\end{proof}

This previous lemma states that we need only prove the result for counting fixed-points over bipartite graphs. We now proceed to a technical lemma.

\begin{MyLem} \label{BSubgraphProp}
Consider a graph $(G,k)$ in $\G_n{\times}\K_n$ where $G = (\I_n,E)$.
Suppose that there exists $P_s$ and $P_c$ disjoint subsets of $\I_n$ such that $k$ is equal to $2$ on $P_s\cup P_c$ and
\begin{itemize}
\item The cardinality of $P_s$ is exactly equal to 3.
\item The cardinality of $P_c$ is greater than or equal to 3.
\item For all $(s,c)$ in $P_s\times P_c$, we have $sc \in E$.
\item For each $s$ in $P_s$ there exists no node $j$ in $\I_n\backslash P_c$ such that $sj \in E$.
\item Every node in $P_c$ has degree less than or equal to 4.
\end{itemize}
If $a$ in $\A_n$ is a fixed point of $G_k$, then $a$ can take only one value over $P_s\cup P_c$.
\end{MyLem}

\begin{proof}
Let $a$ be a fixed point of $G_k$, and suppose there are at least two nodes $i$ and $j$ in $P_c$ such that, $a_i = a_j = \black$. Since $a$ is a fixed point, $a$ is equal to $\black$ on $P_s$, it then follows that $a$ is equal to $\black$ on $P_c$. Similarly, suppose there is at most one node $i$ in $P_c$ such that $a_i$ is $\black$. Since $a$ is a fixed point, $a$ is equal to $\white$ on $P_s$, it then follows that $a_i = \white$, contradicting the assumption.
\end{proof}

Let us define \#bipartite-FIX to be the counting problem \#FIX while restricting an input $<n,G,k>$ to have $G$ being a bipartite graph in $\G_n^b$. We get the following theorem:

\begin{MyThm} \label{SPComplete}
 The counting problem \#bipartite-FIX is \#P-Complete.
\end{MyThm}

\begin{proof}
Clearly \#bipartite-FIX is in \#P, an action configuration can be easily verified to be a fixed-point in polynomial time. We perform a reduction from \#monotone-2DNF that takes $<\phi>$ as input where $\phi$ is a monotone boolean formula in 2-DNF\footnote{A boolean formula is in Disjunctive Normal Form (DNF) if it is a disjunction of conjunction clauses. Formally, let $x_1, \cdots x_n$ be boolean variables, a boolean formula $\phi$ is in DNF-form if $\phi$ is written as
$\phi(x_1, \cdots, x_n) = \bigvee_{i=1}^m  (y_{i,1} \wedge \cdots \wedge y_{i,k_m})$ where for each $m$, $k_m$ is a positive integer, and the literal $y_{.}$ represents either $x_i$ or its negation $\neg x_i$ for some $i$. A k-DNF formula is DNF formula where each clause contains no more than $k$ literals. A boolean formula in DNF is monotone if no clause contains a negation of a variable. } and computes the number of satisfying assignments. This problem is shown to be \#P-Complete (The result follows from \cite{VAD01}). Let $\phi$ be a monotone 2-DNF formula of $m$ clauses and $n$ variables $x_1,\cdots,x_n$.
Namely,
\begin{equation}
  \phi(x_1,\cdots,x_n) = \bigvee_{c = 1}^m (x_{y_c} \wedge x_{z_c}),\nonumber
\end{equation}
where $y$ and $z$ are maps from $\{1,\cdots,m\}$ into $\{1,\cdots,n\}$. Without any loss of generality we will assume that all variables appear in the formula, otherwise we can reduce it to a formula having a fewer number of variables. Moreover, if a clause consists on only one literal $x$, we will write it as $x\wedge x$.

We construct a graph $G(\I_{3(n + 2m + m + 1)},E)$ of size $3(n + 2m + m + 1)$ as follows. We consider $\I_{3(n + 2m + m + 1)}$ and label the players as
\begin{itemize}
	\item $s^1_p,s_p^2,s_p^3$ for $p \in \{1,\cdots,n\}$.
	\item $y^1_c,y^2_c,y^3_c,z^1_c,z^2_c,z^3_c$ for $c \in \{1,\cdots,m\}$.
	\item $b^1_c,b^2_c,b^3_c$ for $c \in \{1,\cdots,m\}$.
	\item $d^1,d^2,d^3$.
\end{itemize}
We now let the edge set E contain the following edges:
\begin{itemize}
 \item $b^l_cy^l_c$ and $b^l_cz^l_c$ for all $l \in \{1,2,3\}$ and $c \in \{1,\cdots,m\}$.
 \item $d^lb^l_c$ for all $l \in \{1,2,3\}$ and $c \in \{1,\cdots,m\}$.
 \item $s^l_py^{l'}_c$ if $y_c = p$ for $p$ in $\{1,\cdots,n\}$, $c$ in $\{1,\cdots,m\}$ and $l,l'$ in $\{1,2,3\}$.
 \item $s^l_pz^{l'}_c$ if $z_c = p$ for $p$ in $\{1,\cdots,n\}$, $c$ in $\{1,\cdots,m\}$ and $l,l'$ in $\{1,2,3\}$.
\end{itemize}
We define $k$ in $\K_{3(n + 2m + m + 1)}$ such that $k$ is equal to $2$ everywhere on $\I_{3(n + 2m + m + 1)}$ except at $d^1$,$d^2$ and $d^3$ where it is equal to $1$.

Let $f$ be an assignment for $\phi$, namely a mapping $f$ from $\{1,\cdots,n\}$ into $\{0,1\}$ such that $\phi(f) = 1$. Define $A_f$ to be the subset of $FIX_{3(n + 2m + m + 1)}(G,k)$ such that for $a$ in $A_f$,
\begin{itemize}
\item $a(y^{l}_c) = \black$ iff $f(y_c) = 1$, for $c$ in $\{1,\cdots, m\}$ and $l$ in $\{1,2,3\}$,
\item $a(z^{l}_c) = \black$ iff $f(z_c) = 1$, for $c$ in $\{1,\cdots, m\}$ and $l$ in $\{1,2,3\}$.
\end{itemize}

We claim three things. First, for every satisfying assignment $sat$ of $\phi$, the cardinality of $A_{sat}$ is equal to $1$. Second, If $0$ is the all zero assignment of $\phi$, the cardinality of $A_0$ is equal to $1$. Third, for every non-satisfying assignment $nsat$ of $\phi$ different than the all zero assignment, the cardinality of $A_{nsat}$ is equal to $8$.

Let $sat$ be a satisfying assignment of $\phi$. For $a$ in $A_{sat}$, all nodes corresponding to the same variable have the same action, and by lemma \ref{BSubgraphProp}, since $a$ is a fixed point, all $s^l_p$ have the same color as the nodes connected to them. Since $sat$ is a satisfying assignment, at least one clause is satisfied. Let that clause be $c_{sat}$. It then follows that since $a$ is a fixed point, $b^l_{c_{sat}}$ is $\black$, and so $d^l$ is $\black$. All other actions are then deterministically set. Therefore, $A_{sat} = \{a\}$.

In the case of the all zero assignment $0$, let $a$ be an element of $A_0$. Again, all nodes corresponding to the same variable have the same action, and by Lemma \ref{BSubgraphProp}, all $s^l_p$ have the same color as the nodes connected to them. The assignment $0$ is non-satisfying, therefore $d^l$ is $\white$. It then follows that all nodes in the graph are $\white$. Therefore $A_{0} = \{a\}$.

Finally, let $nsat$ be a non-satisfying assignment of $\phi$ different than the all zero assignment, and let $a$ be an element of $A_{nsat}$. First we have that all $s^l_p$ have the same color as the nodes connected to them. Since $nsat$ is not the all zero assignment, there exists $c$ such that either $y^l_c$ is $\black$ or $z^l_c$ is $\black$. Pick $l$ in $\{1,2,3\}$. Suppose that $d^l$ is $\black$ then $b^l_c$ is $\black$ and the color of all $b^l_{c'}$ with $c\neq c'$ are deterministically set. Suppose that $d^l$ is $\white$ then $b^l_{c'}$ is $\white$ for all $c'$. Therefore, $|A_{nsat}| = 8$.

For assignments $f$ and $f'$ for $\phi$, clearly if $f \neq f'$ then $A_f \cap A_{f'} = \emptyset$. It also follows from Lemma \ref{BSubgraphProp} that for any fixed point $a$ of $G_k$, there exists an assignment $f$ of $\phi$ such that $a \in A_f$. Indeed all the nodes corresponding to the same variable share the same color. Then consider a set of $n$ nodes corresponding to the different variables (such a set exists since we assumed that all literals appear in the formula), the coloring of those nodes translates to an assignment $f$ of $\phi$ such that $A_f$ contains the fixed point considered.

Let $\#sat$ and $\#nsat$ be respectively the number of satisfying and non-satisfying assignments of $\phi$. Let $F$ be the cardinality of $FIX_{3(n + 2m + m + 1)}(G,k)$, then
\begin{equation}
 \#sat + \#nsat = 2^n \quad \text{ and } \quad \#sat + 8(\#nsat - 1) + 1 = F.\nonumber
\end{equation}
The graph is bipartite and can be constructed in polynomial time, the system of equation can also be solved in polynomial time.
\end{proof}

If we define \#bipartite-CYCLE and \#bipartite-2CYCLE to be respectively the counting problems \#CYCLE and \#2CYCLE while restricting an input $<n,G,k>$ to have $G$ a bipartite graph in $G_n^b$, we arrive to the following corollaries.

\begin{MyCorr}
\#bipartite-CYCLE is \#P-Complete.
\end{MyCorr}
\begin{proof}
Combine Lemma \ref{FunctionOFix} and Theorem \ref{SPComplete}.
\end{proof}

\begin{MyCorr}
\#bipartite-2CYCLE is \#P-Complete.
\end{MyCorr}
\begin{proof}
Combine Lemma \ref{FunctionOFix} and Theorem \ref{SPComplete}.
\end{proof}

The results stated in theorem \ref{CCycle}, \ref{CFix} and \ref{CTwoCycle} hold then by inclusion. We can even claim stronger statements. Notice that we only used thresholds equal to $1$ or $2$. Even more, we could make the construction by restricting the thresholds to be only equal to $2$ everywhere simply by merging $d^1$, $d^2$ and $d^3$ into one node. In some sense, the complexity is not truly coming from the heterogeneity of the threshold, but rather from the threshold rule and the network complexity. General graph structures do indeed incur a good amount of complexity, however ordered graph structure can be fairly tractable. For instance, counting over complete graphs is in FP (see \cite{MyThesis}).

\subsection{Reachability and Counting Predecessors}

We proceed to answer the question of reachability. Given a graph structure $G$, a threshold distribution $k$ and some action configuration $a$, can $a$ be reached from some action configuration? We define the following:

\begin{MyDef}
 The language PRED consists of all 4-tuples $<n,G,k,a>$, where $n$ is a positive integer, $(G,k,a)$ belongs to $\G_n\times\K_n\times\A_n$ with $G_k(a)^{-1} \neq \emptyset$.
\end{MyDef}

We get the following result:

\begin{MyThm}
 PRED is NP-Complete.
\end{MyThm}

\begin{proof}
	First, clearly PRED is in NP. We now perform a reduction from the decision problem 3SAT that takes $<\phi>$ as input where $\phi$ is a boolean formula in 3-CNF\footnote{A boolean formula is in Conjunctive Normal Form (CNF) if it is a conjunction of disjunction clauses. Formally, let $x_1, \cdots x_n$ be boolean variables, a boolean formula $\phi$ is in CNF-form if $\phi$ is written as $\phi(x_1, \cdots, x_n) = \bigwedge_{i=1}^m  (y_{i,1} \vee \cdots \vee y_{i,k_m})$ where for each $m$, $k_m$ is a positive integer, and $y_.$ represents either $x_i$ or its negation $\neg x_i$ for some $i$. For each possible $i$ and $j$, $y_{i,j}$ is called a \emph{literal}. For each $m$, the disjunction of literals $(y_{i,1} \vee \cdots \vee y_{i,k_m})$ is called a clause.  A k-CNF is a CNF formula where each clause contains no more than k literals.} and outputs whether or not there exists a satisfying assignments. This problem is known to be NP-Complete (See \cite{SIP01}). Let $\phi$ be a 3-CNF formula of $m$ clauses and $n$ variables $x_1,\cdots,x_n$.
Namely,
\begin{equation}
  \phi(x_1,\cdots,x_n) = \bigwedge_{c = 1}^m (y_c \vee w_c \vee z_c),\nonumber
\end{equation}
We construct an undirected graph $G(\I_{4n + m + 1},E)$ as follows. We consider $\I_{4n + m + 1}$ and label the players as
\begin{itemize}
	\item $v_p,v'_p,o_p,t_p$ for $1\leq p \leq n$.
	\item $s_c$ for $1\leq c \leq m$.
	\item $u$.
\end{itemize}
We let E contains the following undirected edges:
\begin{itemize}
\item $o_pv_p$, $o_pv'_p$, $t_pv_p$, and  $t_pv'_p$ for all $p$.
\item $s_cv_p$ if and only if $x_p \in \{y_c,w_c,z_c\}$ and $s_cv'_p$ if and only if $\neg x_p \in \{y_c,w_c,z_c\}$ for all $p$ and $c$.
\item $uv_p$ and $uv'_p$ for all $p$. 
\end{itemize}
We define $k$ in $\K_{4n+m+1}$ to be equal to $1$ everywhere on $\I_{4n + m + 1}$ except at $t_1,\cdots,t_m$ where it is equal to $2$.
Finally, we construct $a$ in $\A_{4n+m+1}$ such that $a$ is $\black$ on $\I_{4n + m + 1}$ except at $t_1,\cdots,t_m$ where it $\white$.
We claim that there exists $b$ in $\A_{4n+m}$ such that $a = G_kb$ if and only if $\phi$ is satisfiable. To see that, suppose such a $b$ exists. We have that $v_p$ and $v'_p$ represent the variable $x_n$ and its negation $\neg x_n$ respectively. The node $o_p$ being $\black$ enforces either $v_p$ to be $\black$ or $v'_p$ to be $\black$, $t_p$ being $\white$ enforces either $v_p$ to be $\white$ or $v'_p$ to be $\white$. It follows that $v_p$ and $v'_p$ are of opposite colors in $b$. Finally, $s_c$ being $\black$ enforces that clause $c$ is satisfied. The node $u$ is only needed to ensure connectedness of the graph. To prove the converse, suppose $\phi$ has a satisfying assignment, let $sat$ be such an assignment. Construct $b$ as follows: make $b$ equal to $\black$ everywhere on $\I_{4n + m + 1}$ except on $v_p$ and $v'_p$ for all $p$. Finally, set $b_{v_p} = \black$ if and only if $x_p = 1$ and $b_{v'_p} =  \black$ if and only if $x_p = 0$ for all $p$.
Finally, the construction of the graph is done in polynomial time.
\end{proof}

Given a graph structure $G$, a type distribution $k$ and a configuration $a$, suppose we want to compute the number of configurations $b$ from which $a$ can be reached by applying $G_k$ only once on $b$. We define:

\begin{MyDef}
The counting problem $\#PRED$ takes $<n,G,k,a>$ as input, where $n$ is a positive integer and $(G,k,a)$ and outputs the cardinality of $G_k^{-1}(a)$.
\end{MyDef}

As a corollary from the hardness of PRED, we get:
\begin{MyCorr}
 $\#PRED$ is \#P-Complete.
\end{MyCorr}

However, suppose that we restrict the counting to only the action configurations that are reachable from some action configuration. Specifically, we restrict the counting to only the elements in $PRED$. From this perspective, we are computing the `fan-in' of given configuration.

\begin{MyDef}
The counting problem \#reachable-PRED takes $<n,G,k,a>$ as input, where $n$ is a positive integer, $(G,k,a)$ belongs to $\G_n{\times}\K_n{\times}\A_n$ and $G_k(a)^{-1} \neq \emptyset$ and outputs the cardinality of $G_k^{-1}(a)$.
\end{MyDef}

We get the following result:
\begin{MyThm} \label{CPred}
 \#reachable-PRED  is \#P-Complete.
\end{MyThm}

\begin{proof}
	First, the fact that \#reachable-PRED belongs to \#P is clear. We now perform a reduction from the counting problem \#monotone-2-SAT that takes $<\phi>$ as input where $\phi$ is a monotone boolean formula in 2-CNF\footnote{A boolean formula in CNF is monotone if no clause contains a negation of a variable.} and computes the number of satisfying assignments. This problem is shown to be \#P-Complete (see \cite{VAD01}). Let $\phi$ be a Monotone-2-CNF formula of $m$ clauses and $n$ variables $x_1,\cdots,x_n$.
Namely,
\begin{equation}
  \phi(x_1,\cdots,x_n) = \bigwedge_{c = 1}^m (y_c \vee z_c), \nonumber
\end{equation}
We construct a graph $G(\I_{n+m+1},E)$ as follows. We consider the set $\I_{n+m+1}$ and label the nodes as follows $v_1,\cdots,v_n$, $u_1,\cdots u_m$ and $d$. Construct the edge set $E$ in such a way that $u_cv_p$ belongs to $E$ if and only if $x_p$ appears in clause $c$. Finally we let $dv_p$ belong to $E$ for every $p$. Define $k$ in $\K_{n+m+1}$ to equal to $1$ everywhere on $\I_{n+m+1}$.

Let $a$ be the action configuration in $\A_{n+m+1}$ such that $a$ is equal to $\black$ everywhere. Clearly $a$ is reachable in $(G,k)$ since it is a fixed-point. We claim that the number of configurations preceding $a$, i.e. $|G_k^{-1}a|$ is equal to the number of satisfying assignments for $\phi$. To show that we set up a bijection from $G_k^{-1}a$ into the set of satisfying assignment. Let $b$ be any action configuration in $G_k^{-1}a$, then necessarily $b$ is $\black$ on $u_l$ for all $l$ and on $d$. Any coloring configuration on the $v_p$ nodes that would induce an action configuration in $G_k^{-1}a$ is actually a satisfying assignment, and any satisfying assignment translates as a coloring on the $v_p$ nodes that yields an action configuration in $G_k^{-1}a$.
\end{proof}

We transition to the resilience analysis context in the next section.

\section{Resilience of Networks}

In this section, we revert back to the primary model (see subsection \ref{Model_Dynamics}) where we consider \emph{types} instead of \emph{thresholds}, namely $\Q_n$ instead of $\K_n$. All the needed definitions in this paper including $\K_n$ naturally extend to the set $\Q_n$. Mainly, for $G(I_n,E)$ in $\G_n$ and $q$ in $\Q_n$, we denote by $G_q$ the map from $\A_n$ into $\A_n$ such that for player $i$, $(G_q a)_i = \black$ if and only if (strictly) more than $q_id_i$ players are in $a^{-1}(\black)\cap \N_i$.\\ 

We explain the resilience problem formulation. Given a graph structure $G$ and a positive integer $K$, we suppose that at most $K$ players in the network start playing action $\black$. We want to allocate a type distribution $q$ over the players, so that the dynamics depicted in Proposition~\ref{rule} lead all the agents to play action $\white$ at the limit. We define a \emph{cost function} from $\Q_n$ into the reals $\mathbb{R}$, and the goal is to find a type distribution (recovering the network) that minimizes this cost function. From this perspective, the resilience measure would be this minimal cost of type investment required to recover the network $G$ from a perturbation (from the all $\white$ configuration) of magnitude $K$. In this sense, the lower the resilience measure is for a graph $G$, the more robust $G$ is against perturbations, in that we mean the less costly it is to allocate types to have $G$ recover.

\subsection{The Resilience Measure}
We define $\lone{.}$ to be the map from $\Q_n$ into the reals $\mathbb{R}$ such that, for $q$ in $\Q_n$:
\begin{equation}
 \lone{q} = \sum_{i\in \I_n}{q_i}.\nonumber
\end{equation}
We restrict the analysis in the paper to the map $\lone{.}$.
Let $K$ be a positive integer, we denote by $\A^K_n$ the subset of $\A_n$ such that, $a$ is in $\A^K_n$ if and only if the cardinality of $a^{-1}(\black)$ is at most $K$. We denote respectively by $\white^n$ and $\black^n$ the (constant) action configurations in $\A_n$ mapping each player in $\I_n$ into $\white$ and $\black$ respectively. Recall that for $a$ and $b$ in $\A_n$, we have $b\R_{G_q}a$ if and only if $b = G^m_qa$ for some non-negative integer $m$ (see Section 3). Given a graph $G$ in $G_n$, we define $\Q^{G,K}_n$ to be the subset of $\Q_n$ such that for every $q$ in $\Q^{G,K}_n$ and $a$ in $\A^K_n$, we have $\white^n\R_{G_q}a$ . We define $\mu^K_n(G)$ to be the resilience measure of a graph $G$ with respect to $K$ deviations to be
\begin{equation}
 \mu^K_n(G) = \inf\{\ \lone{q} : q \in \Q^{G,K}_n\ \}.\nonumber
\end{equation}
Note that without any loss of generality, we may assume that for any $q$ in $\Q^{G,K}_n$, $q_i$ is of the form $m/d_i$ for $0 \leq m \leq d_i$.\\

Recall that in Section 2, we assumed that players break ties by playing $\white$. This allows $\white$ to possibly be an always best-response action, and ensures that $\Q^{G,K}_n$ is never empty. Indeed, allocating types of $1$ to all nodes always recover the network.  If this assumption seems to be artificial, we can always recast the resilience problem in terms of integer thresholds modifying the (considered) cost function to minimize to be `roughly' the sum of thresholds normalized by the degree (specifically mapping a threshold distribution $k$ to $\sum_i \frac{k_i - 1}{d_i}$). In this setting, there are no natural restrictions on the thresholds and it is possible to have a node always decide on $\white$ whenever it is allowed to play.

\subsection{Lower Bounds}

We prove lower bounds on the resilience measure.

\begin{MyThm}
 The resilience measure $\mu^K_n$ is greater than or equal to $1$ for all $K \leq n$.
\end{MyThm}

\begin{proof}
Without any loss of generality, we may assume that $K=1$. Let $G(\I_n,E)$ be a graph in $\G_n$, and set $d_{max} = n - 1 - m = \max_{i\in \I_n} d_i$, for some non-negative integer $m < n-2$ since the graph is connected. Let $q$ in $\Q_n$ be given, and let $k$ be the number of zero coordinates in $q$, namely the cardinality of $q^{-1}(0)$. If $k = n$, then at least two nodes $i$ and $j$ with zero $q_i$ and $q_j$ satisfy $ij \in E$, and it would follow that $q \notin \Q^{G,1}_n$. We then have $k<n$. If $k = 0$, then necessarily $\lone{q} \geq \frac{n}{d_{max}} \geq 1$. So we suppose $k\neq 0$

Since $d_{max} = n - 1 - m$, there exists at least one player that is connected to $n - m - 1$ players, or put differently, that leaves $m$ players not connected to it. Suppose that player $i$ with $d_i = d_{max}$ has $q_i = 0$, it would follow that $\lone{q} \geq d_{max}/d_{max} \geq 1$. So assume that at least some player $i$ with $d_i = d_{max}$ has $q_i>0$. Then the players in $q^{-1}(0)$ can be covered by a minimal set $S$ of $p \leq 1 + m$ distinct players with non-zero types, one of them being player $i$ i.e. each player in $q^{-1}(0)$ has at least one of the players in $S$ as a neighbor. If $q$ belongs to $\Q^{G,K}_n$, then every player $i$ with $q_i = 0$ will contribute (independently) at least $1/d_j$ to every $q_j$ with $j$ in $\N_i$. Indeed, if $q$ belongs to $\Q^{G,K}_n$ then each player $j$ connected to $p$ nodes having a type of $0$ will necessarily have a type greater than $\frac{p}{d_j}$. Therefore:
\begin{equation}
 \lone{q} \geq \frac{k}{d_{max}} + \frac{n - p - k}{d_{max}} \geq  \frac{n - (m+1)}{d_{max}} = 1, \nonumber
\end{equation}
where $k/d_{max}$ comes from the contribution of the $k$ zero type players and is taken into account in the types of the players in $S$, and each remaining $n-p-k$ player will contribute at least $1/d_{max}$.
\end{proof}

We show that the bound is tight.

\begin{MyPro}
 This bound is achieved by the star graph $S_n$ for all $n$ and $K$. The star graph $S_n$ is the unique optimal solution for $K>1$. 
\end{MyPro}

\begin{proof}
To show that the bound is achieved by the star network, allocate types on the graph such that the node with degree $n-1$ has a type of $1$ and the rest a type of $0$. We now prove uniqueness for $K>1$.

Let $G(\I_n,E)$ be a graph in $\G_n$, and set $d_{max} = n - 1 - m = \max_{i\in \I_n} d_i$, for some non-negative integer $m < n-2$ since the graph is connected. Let $q$ in $\Q_n$ be given, and let $k$ be the number of zero coordinates in $q$, namely the cardinality of $q^{-1}(0)$. If $m+1 > k$ then:
\begin{equation}
 \lone{q} \geq \frac{n-k}{d_{max}} > 1. \nonumber
\end{equation}
If $m+1 \leq k$, we suppose we can minimally cover those $k$ nodes with $p$ nodes. That is, let $S$ be a subset of $V(G)$ having the smallest cardinality $p$, such that each node of those $k$ nodes is connected to a node in $S$. We have $1 \leq p \leq m+1$.  We suppose $m \neq 0$. If $p < m+1$, we have that:
\begin{equation}
 \lone{q} \geq \frac{k}{d_{max}} + \frac{n-k-p}{d_{max}} \geq \frac{n-p}{d_{max}} > 1. \nonumber
\end{equation}
We then suppose $p = m+1$. If two of the nodes in $S$ are connected by an edge, since $K>1$, one of those nodes (call it $i$) has a type equal to at least $2/d_i$, and so:
\begin{equation}
 \lone{q} \geq \frac{k}{d_{max}} + \frac{1}{d_{max}} + \frac{n-k}{d_{max}} > 1. \nonumber
\end{equation}
Finally, we suppose none of the nodes in $S$ are neighbors, then each can have a maximum degree of $n-1-m-1$, since each node is not connected to the $m$ others and at least one of the zeros is not connected to it. It follows that:
\begin{equation}
 \lone{q} \geq \frac{k}{d_{max}-1} + \frac{n-k}{d_{max}} > 1. \nonumber
\end{equation}

Finally, if $m = 0$, then one node has necessarily a type equal to $1$. If $k < n-1$, then necessarily $\lone{q} > 1$. The remaining case is $k=n-1$, and the only possible graph is the star graph since $q$ belongs to belongs to $\Q^{G,K}_n$ and hence no nodes of type $0$ should be neighbors.
\end{proof}

For $K=1$, the complete graph is also an optimal solution.

\subsection{Upper Bounds}
We prove upper bounds on the resilience measure.
\begin{MyThm}
 The resilience measure $\mu^K_n$ is less than or equal to $n/2$ for all $K \leq n$.
\end{MyThm}

\begin{proof}
Without any loss of generality, we may assume that $K=n$. In this case, the players start by playing the action configuration $\black^n$, and we need all the players to play $\white$ at the limit. We only need to find a $q$ is in $\Q^{G,n}_n$ with $\lone{q} \leq n/2$.  To this end, impose a strict order relation $<$ on $\I_n$, such that for every $i$ and $j$ in $\I_n$,
\begin{equation} \label{strictOrder}
 i < j \quad \text{if} \quad d_i < d_j. 
\end{equation}
We note that that the statement in (\ref{strictOrder}) is not an \emph{only if} statement. The case where $d_i = d_j$ is taken care of by the fact that $<$ is a strict order.
We construct $q$ as follows: for $i$ in $\I_n$, set:
\begin{equation}
q_i = \sum_{j\in \N_i : j < i} d_i^{-1}. \nonumber
\end{equation}
We are iterating over all edges $\{i,j\}$, and adding $d^{-1}_k$ to $q_k$ of player $k$ in $\{i,j\}$ that has the highest degree, or if the degrees are equal that is the tie breaker set by the order relation $<$.

Then, the type distribution $q$ is in $\Q^{G,n}_n$. To show that, we set up an order preserving bijection $r^{-1}$ from $(\I_n,<)$ into $(\{1,\cdots,n\},>)$, and we refer to player $r(k)$ as simply player $k$, for $k$ in $\{1,\cdots,n\}$. So player $1$ refers to the `largest' player. We claim that player $k$ will be playing $\white$ after applying $G^k_q$. We prove this by induction. Player $1$ will have necessarily have $q_{r(1)} = 1$, and so will necessarily play $\white$ when we apply $G_q$. Suppose the statement is true for player $k$, we show that the statement is true for player $k+1$. After $G^k_q$, all players $k'$ with $k'\leq k$ are playing $\white$. Assume node $k+1$ has degree $d_{r(k+1)}$, and suppose it is connected to $m$ players `smaller' than it, then it has $q_{r(k+1)} = md_{k+1}^{-1}$, and so it needs more than $m$ neighbors $\black$ to play $\black$, but all the players that are `larger' than it are playing $\white$, so player $k+1$ will play $\white$ when $G_k$ is applied one more time.

Finally, $\lone{q} \leq n/2$. To prove that, each node $i$ has degree $d_i$, and so can contribute no more than $d_id^{-1}_i=1$ to $\lone{q}$. If we give each player $i$ a type $q_i = 1$, each edge is then counted twice in the summation, then:
\begin{align}
	n = \sum_{i \in \I_n}\sum_{j \in \N_i}{\frac{1}{d_i}} &= \sum_{ij \in E}{\frac{1}{d_i} + \frac{1}{d_j}}\nonumber\\
				&= \sum_{ij \in E}{\frac{1}{d_i} \wedge \frac{1}{d_j} + \frac{1}{d_i} \vee \frac{1}{d_j}}. \nonumber
\end{align}
But by construction, we know that
\begin{equation}
 \lone{q} = \sum_{ij \in E}{\frac{1}{d_i} \wedge \frac{1}{d_j}}. \nonumber
\end{equation}
The result follows since ${d^{-1}_i} \wedge {d^{-1}_j} \leq {d^{-1}_i} \vee {d^{-1}_j}$.
\end{proof}

We show that the bound is tight for large $K$.

\begin{MyPro}
 This bound is achieved by the cycle graph $R_n$ for all $n$ and $K \geq \lceil n/2 \rceil$. 
\end{MyPro}

\begin{proof}
It would be enough to show that the bound is achieved for $K = \lceil n/2 \rceil$ by the 2-regular connected graph.

Let us consider the case where $n$ is even. Let $q$ be a type distribution in $\Q^{G,n}_n$ where $G$ is a 2-regular connected graph. We begin by claiming, that for any distinct $i$ and $j$ in $\I_n$, if both $q_i$ and $q_j$ are $0$, then for any (vertex) path $P$ from $i$ to $j$ there exists a node $k$ in $\I_n$ (distinct from $i$ and $j$) such that $k$ lies on the (vertex) path and $k$ has a type equal to $1$. To prove that, we pick two nodes $i$ and $j$, and consider a path $(i,i_1,\cdots,i_m,j)$ from $i$ to $j$. Since $i$ and $j$ are distinct, then $m \leq n - 2$. We now suppose there are no nodes having a type equal $1$ along the path, then all the types along the path are less than 1. Suppose $m$ is even and consider some $a$ in $\A^K_n$ such that $a_i, a_{i_2}, \cdots, a_{i_m}$ are all $\black$, that is possible since $m/2 + 1 \leq n/2$, applying $G_q$ once on $a$ yields $a_{i_1}, \cdots, a_{i_{m-1}},a_j$ are all $\black$, applying it one more time yields $a_i, a_{i_2}, \cdots, a_{i_m}$ are all $\black$ again. Therefore, there exists $a$ in $\A^K_n$ such that $a$ is not in the same equivalence class as $\black^n$, i.e. $(\black^n,a) \notin \R_{G_q}$. By a similar argument, the same thing holds in the case where $m$ is odd.

The second claim is that we need at least one node in the graph to have a type equal to $1$, otherwise not all nodes will play $\white$ at the limit. To see that, suppose no such node exists. We construct $a$ in $\A^K_n$ such that no two neighboring players have the same action. Applying $G_q$ once on $a$ makes at least all the players that were playing $\white$ play $\black$, and applying it one more time makes the initial players that were $\black$ play $\black$ again. Therefore, $a$ does not belong to $\A^K_n$.

We cannot have more than $n/2$ nodes with type $0$ in the graph, otherwise necessarily two nodes with type $0$ will be connected. Then, suppose we have $k$ nodes having type $0$, we get at least $k$ nodes having type $1$ and the rest is $1/2$. If we sum the types, we get $n/2$.

For the case where $n$ is odd, following a similar argument, we establish that we should have at least one node with a type equal to $1$ in the network.
Then, for any two disjoint players having a type of $0$, every (vertex) path connecting the two players should contain at least one node having a type equal to $1$.
\end{proof}

\subsection{The Resilience of Cycle Graphs and Complete Graphs}

We derive the resilience of path graphs, cycle graphs and complete graphs.

\begin{MyPro} \label{PathMeasure}
 The path graph $P_n$ of size $n$ has a resilience measure $\mu^K_n(R_n)$ equal to $(n - 1 -\lfloor\frac{n-1}{2K+1}\rfloor)/2$ for $K <\lceil n/2 \rceil$.
\end{MyPro}

\begin{proof}
Since $K <\lceil n/2 \rceil$, the nodes having degree 1 in the path graph have necessarily a type equal to $0$. Suppose no node has a type equal to $1$, then types are either equal to $0$ or equal to $1/2$. Given that $K$ players start playing $\black$, every two players with types equal to $0$ should be separated by at least $2K$ players having type equal to $1/2$. If $k$ is the number of players with type equal to $0$, then $(n-k)/(k-1) \geq 2K$ or equivalently $k \leq \frac{n-1}{2K+1} +1$. The maximum number of type $0$ nodes is then $\lfloor\frac{n-1}{2K+1}\rfloor + 1$. Finally, suppose a node with type $1$ exists, we may split the path into two smaller paths, and it can be checked that is yield only a suboptimal type distribution.
\end{proof}

\begin{MyPro}
 The cycle graph $R_n$ of size $n$ has a resilience measure $\mu^K_n(R_n)$ equal to $(n - \lfloor\frac{n}{2K+1}\rfloor)/2$ for $K <\lceil n/2 \rceil$.
\end{MyPro}

\begin{proof}
We first show that there exists an optimal allocation of types that is nowhere equal to $1$ on $\I_n$. Suppose some node has a type equal to $1$, then we can delete that node and obtain a path. The optimal allocation in that path is one with no node having a type equal to $1$. Therefore, we can only have one node (call it $i$) having type equal to $1$ if ever. If $i$ is connected to a node with type equal to $0$, we can replace the two types by $1/2$ while keeping a type distribution in $\Q_n^K$. If $i$ is connected to two nodes with types equal to $1/2$, we can swap the type of one of the neighbors with the type of $i$ while keeping a type distribution in $\Q_n^K$. We can keep `moving' the type equal to $1$ till its corresponding node is connected to a node having a type of $0$. We then apply the previous argument.

This said, an optimal allocation need not create a node having type equal to $1$. Let us assume we have such an allocation, and let us determine the maximum number of type $0$ nodes we can include. Necessarily one node of type $0$ exists. The result follows by applying Proposition \ref{PathMeasure} on a path of length $n+1$.
\end{proof}

\begin{MyPro}
 The complete graph $K_n$ of size $n$ has a resilience measure $\mu^K_n(K_n)$ equal to $\frac{K(K-1)/2 + K(n-K)}{n-1}$ for all $K$.
\end{MyPro}

\begin{proof}
Define $d$ to be equal to $n-1$. We consider the type distribution $\hat{q}$ such that $\hat{q}$ equals $K/d$ on exactly $n-K$ players, and is injective when restricted to the remaining players taking values in $\{0,1/d,\cdots,(K-1)/d\}$. We claim that $\hat{q}$ belongs to $\Q^K_n$. To see that, we argue as follows. Each node is initially connected to either $K$ or $K-1$ neighboring nodes playing $\black$. Therefore, the type distribution $q$ equal to $K/d$ everywhere is in $\Q_n^K$. This distribution is however not optimal. We can allow $K$ nodes in $\I_n$ to have types equal to $(K-1)/d$, while the rest have types of $K/d$ and keep the distribution in $\Q_n^K$. Let that distribution be $q'$. Given any initial configuration in $\A_n^K$, after one application of $G_{q'}$, only $K$ of the nodes can play $\black$. The $n-K$ nodes having types equal to $K/d$ will always play $\white$. Reapplying the same argument on the remaining nodes, we can make $K-1$ of them have a type of $(K-2)/d$. We keep on iterating till we get the type distribution $\hat{q}$ such that $q$ equals $K/d$ on exactly $n-K$ players, and is injective when restricted to the remaining players taking values in $\{0,1/d,\cdots,(K-1)/d\}$. 
 
We now argue that we cannot find a type distribution $q$ that belongs to $\Q^K_n$ such that $\lone{q} < \lone{\hat{q}}$. Suppose we can, then necessarily an action configuration in $\A_n^K$ is a fixed-point contradicting the fact that $q$ belongs to $Q^K_n$ (see \cite{MyThesis} for more insight on Complete Graphs).
\end{proof}

To end this section, we give a small piece of insight. High degree nodes lower the resilience measure in the graph. One manifestation of this fact lies in the examples that meet the bounds. However, if we consider the complete graph, it has a resilience measure of $1$ for $K=1$ that grows to $n/2$ for $K=n$. This said, although high degree nodes increase the resilience of a network, having a large number of high degree nodes in the network makes the network more \emph{fragile} against large perturbation, by making it more costly to ensure its recovery.

\section{Conclusion}
In this paper, we considered a linear threshold model where agents are allowed to switch their actions multiple times. We focused on characterizing the behavior of the dynamics.

We established that in the limit, the agents in the network cycle among action profiles. We studied the lengths of such cycles, and the required number of time steps needed to reach such cycles. In particular, we showed that for \emph{any graph structure} and \emph{any threshold distribution} over the agents, such cycles consist of a most two action profiles. Namely, in the limit, each agent either always plays one specific action or switches action at every single time step.
We also showed that over \emph{all graph structure} (of size $n$) and \emph{all threshold distributions} no more than $cn^2$ time steps are required to reach such cycles, where $c$ is some integer. We also improve convergence time results to be not more than $n$ steps when the underlying graph is either an even cycle graph or a tree. Our methods follow a combinatorial approach, and are based mainly on two techniques: transforming the general graph structure into a bipartite structure, and transforming the parallel dynamics on this bipartite structure into sequential dynamics.

We also studied the problem of counting the number of cycles (fixed-points and non-degenerate cycles), the number of fixed-points and the number of non-degenerate cycles. We showed that those counting problems are \#P-Complete. We further showed that deciding whether an action profile is reachable is NP-Complete and that counting the number of predecessors (preceding action configurations) of a reachable action configuration is \#P-Complete.

Finally, in the setting of resilience of networks, we defined a measure $\mu^K$ that captures the minimal cost of threshold investment required to recover the network $G$ from a perturbation of magnitude $K$, whereby we suppose that $K$ agents will initially deviate from action $\white$ and  play action $\black$. We show that this measure is lower-bounded by $1$, and that it is upper-bounded by $n/2$, where $n$ is the size of the network. We finally provide an interpretation of how this measure varies with respect to the network structures. High degree nodes add resilience to the network, however too many high degree nodes can make the network fragile against strong perturbations.

There are several questions that could be undertook as future work, however we keep the list brief. We consider each section, and provide one or two questions that could be further developed. In Section 4, it would be interesting to provide a characterization of the initial configurations that lead to non-degenerate cycles and those that lead to fixed points. In Section 5, we believe we can improve the convergence time bound for general graphs to be linear in the size of the network. In Section 6, it would be interesting to characterize subclasses of $\G_n {\times} \K_n$ where the counting problems are tractable and devise approximation algorithms if possible for the hard cases. We need to compute resilience measures for more graph structures in Section 7, and mainly devise a systematic way to compute the measure. We may further consider different cost functions on the type distribution.

\begin{singlespace}

\end{singlespace}

\newpage
\appendix

\section{On Function and Counting Problems} \label{ComplexityAppendix}

We assume familiarity with the basic concepts in Complexity Theory and proceed to function/counting problems. We refer the reader to \cite{SIP01} for a thorough exposition of the basics. Function problems would roughly include problems of the form: given some input $w$, compute $f(w)$. Given an alphabet $\Sigma = \{0,1\}$, we refer to $\Sigma^*$ as the set of all binary strings constructed from elements of $\Sigma$.

\begin{MyDef}
 A (binary) relation $R \subset \Sigma^* \times \Sigma^*$ is said to polynomial bounded if and only if there exists a real polynomial $p$ such that for all $x$ and $y$ in $\Sigma^*$, if $xRy$ then $|y| \leq p(|x|)$.
\end{MyDef}

Since functions can be formally thought of as relations, a function $f$ is then polynomial bounded if and only if there exists a real polynomial $p$ such that for every $x$ in domain of $f$, $|f(x)| \leq p(|x|)$.

\begin{MyDef}[Function P]
 A polynomial bounded (binary) relation $R \subset \Sigma^* \times \Sigma^*$ is in $FP$ if and only if there exists a polynomial time deterministic Turing machine such that given an $x$ in $\Sigma^*$, if there exists some $y$ in $\Sigma^*$ where $xRy$, it outputs such a string $y$, else it outputs `no'.
\end{MyDef}

The terminology of function is a bit misleading in the sense that given an $x$, several $y$ might satisfy $xRy$.

\begin{MyDef}[Counting Problems]
 Given a (binary) relation $R \subset \Sigma^* \times \Sigma^*$, the counting problem $\#R$ associated with $R$ is the function from $\Sigma^*$ into $\Zgeqz$ such that for $x$ in $\Sigma^*$, $\#R(x) = |\{y \in \Sigma^* : xRy \}|$.
\end{MyDef}

We now define the class $\#$P.

\begin{MyDef}[$\#$P]
 Given a polynomial bounded (binary) relation $R \subset \Sigma^* \times \Sigma^*$, the counting problem $\#R$ is in $\#P$ if and only if there exists a deterministic Turing machine $M$ that can decide $R$ in polynomial time.
\end{MyDef}

Put differently, the class $\#$P is the class of counting problem associated with languages in NP. For a natural alternate definition, we refer the reader to \cite{VAL01}.\\

The notion of reduction that will be used is one by oracles (as described in \cite{VAL01}). An oracle can be thought of a black-box that can answer queries in one time step. An oracle Turing machine is Turing machine with a query tape, an answer tape, and some working tape. To consult the oracle, the Turing machine (TM) prints a string on the query tape and, on going into a special query state, an answer is returned in unit time on the answer tape, and a special answer state is entered.

\begin{MyDef}
 An oracle Turing machine is said to be in $FP$, if and only if for all polynomial bounded oracles, it behaves like a machine in $FP$.
\end{MyDef}

If $M$ is a class of oracle TMs, and $f$ a polynomial bounded function, then we denote the class of function that can be computed by oracle TMs in $M$ with oracles for $f$, by $M^f$.

\begin{MyDef}
	A polynomial bounded function $f$ is $\#$P-Hard if and only if $\#$P $\subset$ $FP^f$. A polynomial bounded function $f$ is $\#$P-Complete if and only if $\#$P $\subset$ $FP^f$ and $f$ belongs to $\#$P.
\end{MyDef}

In other words, $f$ is \#P-Hard if and only if for every function in $\#$P, there exists an oracle Turing machine from FP with oracle for $f$ that can compute it. Given a counting problem \#R, to prove that \#R is \#P-hard, it would be enough to find a problem \#Q that is \#P-hard, and then reduce \#Q to \#R. Given an instance $q$ in \#Q, reduction should be done in such a way that the construction of the instance $r$ in \#R is performed in polynomial time, and the output for $q$ should be computed from the output of $r$ in polynomial time.

\end{document}